\documentclass[12pt]{article}
\usepackage{amsmath,amsthm}
\usepackage{bm,amssymb,mathtools}
\usepackage[ruled,vlined]{algorithm2e}
\usepackage{graphicx,psfrag,epsf}
\usepackage{enumerate}
\usepackage[round,authoryear]{natbib}
\usepackage{url} 
\usepackage[colorlinks,citecolor=blue,urlcolor=blue,linkcolor = blue,backref=page]{hyperref}
\usepackage{float}
\usepackage{siunitx}
\usepackage{subcaption}
\usepackage{multirow}
\usepackage{booktabs}
\usepackage{tikz}
\usepackage{longtable}
\usepackage{diagbox}
\newtheorem{theorem}{Theorem}

\title{Minimum Density Power Divergence Estimation for the Gamma Distribution with Applications to Robust Rainfall Modeling}
\author{Arnab Hazra\footnote{Email: ahazra@iitk.ac.in} \vspace{2mm} \\
Department of Statistics and Data Science, \\ Indian Institute of Technology Kanpur, Kanpur 208016, India \vspace{2mm}
}

\date{}
\begin{document}

\maketitle
\vspace{3mm}
\begin{abstract} \vspace{2mm}
\noindent 
Statistical modeling of rainfall amounts is of considerable importance in meteorology, hydrology, and agriculture. The gamma distribution remains one of the most popular choices for modeling rainfall data due to its flexibility and ability to capture the skewness of rainfall observations. Rainfall datasets often contain atypical observations due to measurement errors and extreme weather events, making maximum likelihood estimation (MLE) highly sensitive to contamination. In this paper, we develop a robust estimation framework for the two-parameter gamma distribution based on the minimum density power divergence estimator (MDPDE). Explicit estimating equations are derived, and several theoretical properties of the proposed estimators are established. In particular, closed-form expressions for the asymptotic covariance matrix are obtained, and robustness is investigated through influence function analysis and asymptotic relative efficiency. The finite-sample performance of the estimators is examined through simulation studies under both pure and contaminated gamma models. The proposed methodology is further implemented to analyze detrended areally weighted monsoon rainfall data from the 36 meteorological subdivisions of India for the period 1951--2014. The results demonstrate that the MDPDE provides a useful compromise between robustness and efficiency, yielding more stable inference than MLE in the presence of outliers while maintaining high efficiency for uncontaminated data.
\end{abstract}


\section{Introduction}
\label{sec:introduction}
Statistical modeling of rainfall data has long been an important area of research in agrometeorology. Rainfall amounts are typically positively skewed and often exhibit a substantial right tail, motivating the use of probability distributions supported on the positive real line. A wide variety of parametric models have been proposed in the literature, including the exponential distribution \citep{todorovic1975stochastic, burgueno2005statistical, burgueno2010statistical, hazra2018bayesian}, Gamma distribution \citep{barger1949evaluation, mooley1968application, husak2007use, krishnamoorthy2014small, martinez2019precipitation}, lognormal distribution \citep{kwaku2007characterization, mandal2015estimation, adham2016water}, Weibull distribution \citep{duan1995comparison, burgueno2005statistical, lana2017rainfall}, Pearson Type-V and Type-VI distributions \citep{hanson2008probability, khudri2013determination, mayooran2014statistical}, log-logistic distribution \citep{fitzgerald2005analysis, sharda2005modelling}, generalized exponential distribution \citep{madi2007bayesian, kazmierczak2015suitability, hazra2025minimum}, and generalized Gamma distribution \citep{sharma2010use, mandal2015estimation, mamoon2017selection}. The adequacy of these probability models has been investigated using a variety of goodness-of-fit measures and model selection criteria. Commonly employed procedures include the chi-square test \citep{barger1949evaluation, mooley1973gamma, kwaku2007characterization}, the Kolmogorov--Smirnov test \citep{sharma2010use, hazra2014modelling, al2014frequency}, the variance ratio test \citep{mooley1973gamma, hazra2017note}, and the Anderson--Darling test \citep{sharma2010use}, while information-based criteria such as the Akaike Information Criterion have also been widely used for model comparison \citep{villarini2012development}. Despite the availability of numerous competing distributions and formal model selection tools, it is common practice in the agro-meteorological literature to identify a scientifically plausible probability model and fit it to the rainfall data for subsequent inference and prediction.


The gamma distribution is arguably the most widely used probability model for rainfall amounts and wet-period precipitation data \citep{martinez2019precipitation, bhowmik2026bayesian}. Its popularity stems from its flexibility in representing a wide range of positively skewed distributions through only two parameters. Depending on the value of the shape parameter, the gamma distribution can accommodate varying degrees of skewness and dispersion commonly observed in rainfall datasets. Since rainfall amounts are nonnegative and often exhibit a long right tail, the gamma distribution provides a natural probabilistic framework for modeling monthly, seasonal, and annual precipitation totals. Consequently, it has been extensively employed in hydrological frequency analysis, drought assessment, stochastic weather generation, and climate variability studies \citep{barger1949evaluation, mooley1968application, mooley1973gamma, husak2007use, krishnamoorthy2014small}. Its mathematical tractability further facilitates parameter estimation, probabilistic forecasting, and the derivation of important hydrological quantities of practical interest. The successful application of the gamma distribution to rainfall data has been documented across diverse climatic regions and temporal scales. For example, \cite{barger1949evaluation} demonstrate its usefulness for describing precipitation characteristics, while \cite{mooley1968application} and \cite{mooley1973gamma} show that gamma-based models provide an adequate fit to Indian rainfall data in many situations. More recent studies have continued to advocate the gamma distribution owing to its ability to capture the variability inherent in precipitation processes and its compatibility with standard statistical inference procedures \citep{husak2007use, krishnamoorthy2014small, hazra2024robust, bhowmik2026bayesian}. As a result, the gamma distribution is frequently adopted as a benchmark model in rainfall studies, against which alternative probability distributions are often compared. Given its widespread acceptance and practical relevance, developing robust inferential procedures for the gamma distribution is of considerable importance, particularly in the presence of outlying observations and extreme rainfall events that may adversely affect classical likelihood-based estimation methods.


Maximum likelihood estimation (MLE) remains the most widely used method for parameter estimation across a broad range of scientific disciplines, including agro-meteorology, owing to its attractive large-sample properties such as consistency, asymptotic normality, and asymptotic efficiency. However, it is well known that the MLE can be highly sensitive to data contamination and may be severely affected even by a small number of outlying observations, whether arising from recording errors or genuine extreme events \citep{strupczewski2005robustness, strupczewski2007robustness}. Consequently, inference concerning the central part of the distribution, which is often of primary interest in rainfall modeling (except in studies focusing specifically on extremes), can become unreliable. To overcome this limitation, \cite{basu1998robust} introduce a robust estimation framework based on the minimum density power divergence estimator (MDPDE). The MDPDE is obtained by minimizing a discrepancy measure known as the density power divergence (DPD) between the assumed parametric model and the underlying data-generating distribution. For a tuning parameter $\alpha \geq 0$, the DPD between two densities $g$ and $f$ is defined as
\begin{equation}\label{eq:dpd}
    d_\alpha(g,f) = \displaystyle \left\{\begin{array}{ll}
    \displaystyle \int  \left[f^{1+\alpha}(y) - \left(1 + \alpha^{-1}\right)  f^\alpha(y) g(y) + 
\alpha^{-1} g^{1+\alpha}(y)\right] dy, & {\rm ~for} ~\alpha > 0,\\
	\displaystyle \lim_{\alpha\rightarrow 0}d_\alpha(g, f) = \int g(y) \left[\log g(y) - \log f(y) \right] dy, & {\rm ~for} ~\alpha = 0.  
\end{array}\right.
\end{equation}
For $\alpha=0$, the density power divergence reduces to the classical Kullback--Leibler divergence. Considering a parametric model with density function $f_{\bm{\theta}}$, where $\bm{\theta}$ denotes the unknown parameter vector, and let $g$ represent the true data-generating density, the minimum density power divergence estimator is obtained by minimizing $d_\alpha(g, f_{\bm{\theta}})$ with respect to $\bm{\theta}$. A notable feature of this framework is that, when $\alpha=0$, the resulting estimator coincides with the maximum likelihood estimator. Thus, the MDPDE family contains the MLE as a special case and extends it by introducing the tuning parameter $\alpha$. As $\alpha$ increases from zero, the estimator becomes progressively less sensitive to atypical observations and contamination, thereby achieving enhanced robustness at the cost of some loss in statistical efficiency under the assumed model. An important advantage of the MDPDE over several other divergence-based estimation procedures \citep{beran1977minimum, basu1994minimum} is that it does not require nonparametric density estimation or smoothing, thereby greatly simplifying implementation and computation in practice \citep{seo2017extreme}. As a result, the methodology has found successful applications in a variety of fields \citep{gajewski2004correspondence, yuan2008partial, seo2017extreme}. In the context of rainfall modeling, robust divergence-based inference has previously been explored for some commonly used rainfall distributions \citep{hazra2024robust, hazra2025minimum}. Specifically, \cite{hazra2024robust} discuss MDPDE for four probability distributions, including the two-parameter gamma distribution; however, the authors use MDPDE only for parameter estimation and do not provide a detailed analytical or numerical exposition of the theoretical properties of MDPDE. 

In this paper, we develop a comprehensive MDPDE framework for the gamma distribution. We begin by deriving the theoretical properties of the proposed estimators, including their consistency and asymptotic normality. Explicit expressions are obtained for the estimating equations, the sensitivity and variability matrices, and the asymptotic covariance matrix of the MDPDE. We further investigate the robustness properties of the proposed estimators through an influence function analysis, examining the behavior of the influence functions for both the shape parameter $a$ and the rate parameter $b$ across different choices of the tuning parameter $\alpha$. In addition, we study the efficiency--robustness trade-off by evaluating the asymptotic relative efficiencies of the MDPDEs compared with the maximum likelihood estimators. Subsequently, we investigate the effect of the tuning parameter on the fitted gamma distribution and discuss a data-driven procedure for selecting an optimal value of $\alpha$ based on minimizing the empirical Cram\'er--von Mises (CVM) distance, following the ideas of \cite{fujisawa2006robust, hazra2024robust} and \cite{hazra2025minimum}. Extensive simulation studies are conducted to assess the finite-sample performance of the proposed estimators under both pure and contaminated gamma models. The MDPDE is compared with the classical maximum likelihood estimator and several competing estimation procedures commonly used for the gamma distribution. Finally, the methodology is applied to detrended areally weighted monsoon rainfall data from the 36 meteorological subdivisions of India during the period 1951--2014. The resulting analyses illustrate the practical advantages of the proposed robust estimation framework and demonstrate its ability to yield stable, reliable inference in the presence of atypical rainfall observations. Overall, in this paper, we mostly follow the steps in the exploration of MDPDE for a generalized exponential distribution in \cite{hazra2025minimum}, while keeping the main focus on the gamma distribution.

The remainder of the paper is organized as follows. Section \ref{sec:methodology} discusses the MDPDE framework for the gamma distribution. In Section \ref{sec:properties}, we derive its theoretical properties, including the asymptotic distribution, influence function, and asymptotic relative efficiency. A simulation study is presented in Section \ref{sec:simulation} to investigate the performance of the proposed estimators under both pure and contaminated gamma models. In Section \ref{sec:data_application}, we apply the proposed methodology to the detrended areally weighted monsoon rainfall data from the 36 meteorological subdivisions of India and discuss the resulting inferences. Finally, Section \ref{sec:conclusion} concludes the paper with a summary of the main findings and directions for future research.

\section{Methodology} \label{sec:methodology}
In MDPDE, the parameter estimates are obtained by minimizing the density power divergence $d_\alpha(g,f)$ in (\ref{eq:dpd}) over the parameter space. The divergence involves a tuning parameter $\alpha\geq 0$, and for $\alpha=0$, the DPD coincides with the Kullback--Leibler divergence. For the gamma distribution, let the parameter vector be denoted by $\bm{\theta}=(a,b)' \in \Theta=\mathbb{R}^{+}\times\mathbb{R}^{+}$, where $a$ and $b$ denote the shape and rate parameters, respectively, i.e., the density function is given by
\begin{equation}\label{eq:gamma_density}
f_{\Gamma}(y;a,b) =
\frac{b^{a}}{\Gamma(a)}
y^{a-1} \exp(-by),
\qquad y>0.
\end{equation}
For parameter estimation, the DPD is computed as $d_\alpha(g,f_{\Gamma}(\cdot;\bm{\theta}))$, where $g$ denotes the true density function and $f_{\Gamma}$ is the gamma density given in \eqref{eq:gamma_density}. As noted earlier, the MLE is a special case of the MDPDE corresponding to $\alpha=0$. Since the MLE is highly sensitive to outlying observations, the Kullback--Leibler divergence is often unsuitable for robust inference, motivating the use of the density power divergence. Following \cite{basu1998robust}, we define the minimum DPD functional (bivariate) $\bm{T}_\alpha(\cdot)$ by
\begin{equation}\label{eq:dpd2}
d_\alpha(g,f_{\Gamma}(\cdot;\bm{T}_\alpha(G))) = 
\min_{\bm{\theta}\in\Theta}
d_\alpha(g,f_{\Gamma}(\cdot;\bm{\theta})),
\end{equation}
where $G$ denotes the true distribution and $g$ is the density function corresponding to $G$. Thus, $\bm{T}_\alpha(G)$ is the parameter vector that provides the best gamma approximation under $G$.

In practice, however, the true density $g$ is unknown; therefore, an estimator of $g$ is ordinarily required to minimize the DPD. Unlike several alternative divergence-based estimation procedures, the DPD family proposed by \cite{basu1998robust} has the important advantage that it does not require any nonparametric smoothing for estimating $g$. This follows from the representation
\begin{eqnarray}\label{eq:dpd3}
 && d_\alpha(g,f_{\Gamma}(\cdot; \bm{\theta})) \\
\nonumber     &=&\displaystyle \left\{\begin{array}{ll}
    \displaystyle \int f^{1+\alpha}_{\Gamma}(y; \bm{\theta}) dy - \left(1 + \alpha^{-1}\right) \text{E} \left[ f^{\alpha}_{\Gamma}(Y; \bm{\theta}) \right] + 
\alpha^{-1} \text{E}\left[ g^{\alpha}(Y)\right] & {\rm if} ~\alpha > 0,\\
	\displaystyle \text{E}\left[\log g(Y)\right] - \text{E}\left[\log f_{\Gamma}(Y; \bm{\theta})\right] & {\rm if} ~\alpha = 0,  
\end{array}\right.
\end{eqnarray}
where $\text{E}[\cdot]$ denotes expectation with respect to $g$. The two quantities $\text{E}[g^\alpha(Y)]$ and $\text{E}[\log g(Y)]$ are free of $\bm{\theta}$ and hence can be ignored in the optimization problem (\ref{eq:dpd2}). The remaining expectations, namely $\text{E}[f_{\Gamma}^{\alpha}(Y;\bm{\theta})]$ and $\text{E}[\log f_{\Gamma}(Y;\bm{\theta})]$, can be estimated directly using empirical averages and therefore no nonparametric estimation of $g$ is required.

Suppose the random sample we observe is denoted by $Y_1,\ldots,Y_n$. Consequently, the MDPDE of $\bm{\theta}$ corresponding to tuning parameter $\alpha$ is 
\begin{equation}\label{eq:dpd4}
\widehat{\bm{\theta}}_\alpha
=
\arg\min_{\bm{\theta}\in\Theta}
H_{\alpha,n}(\bm{\theta}),
\end{equation}
where
$
H_{\alpha,n}(\bm{\theta}) =
n^{-1}
\sum_{i=1}^{n}
V_\alpha(\bm{\theta};Y_i),
$
with
\begin{equation}\label{eq:dpd5}
   V_\alpha(\bm{\theta}; y)  = \displaystyle \left\{\begin{array}{ll}
    \displaystyle \int f^{1+\alpha}_{\Gamma}(y; \bm{\theta}) dy - \left(1 + \alpha^{-1}\right) f^{\alpha}_{\Gamma}(y; \bm{\theta}),  & {\rm ~~if} ~\alpha > 0,\\
	\displaystyle - \log f_{\Gamma}(y; \bm{\theta}), & {\rm ~~if} ~\alpha = 0.  
\end{array}\right.
\end{equation}
For $\alpha=0$,
$$
\widehat{\bm{\theta}}_0
=
\arg\min_{\bm{\theta}\in\Theta}
\frac{1}{n}
\sum_{i=1}^{n}
\left[
-\log f_{\Gamma}(Y_i;\bm{\theta})
\right],
$$
which is identical to the maximum likelihood estimator $\widehat{\bm{\theta}}_{\rm ML}=(\widehat a_{\rm ML},\widehat b_{\rm ML})^\top$.

To implement the optimization routine in (\ref{eq:dpd4}), we first derive an explicit expression for $V_\alpha(\bm{\theta};y)$. For $\alpha=0$, it follows immediately from the gamma density. For $\alpha>0$, we obtain
\begin{eqnarray}\label{EQ:dpd5_simple}
V_\alpha(\bm{\theta};y)
&=&
\frac{
b^\alpha
\Gamma\left((1+\alpha)a-\alpha\right)
}
{
\Gamma(a)^{1+\alpha}
(1+\alpha)^{(1+\alpha)a-\alpha}
}
- \left(1+\alpha^{-1}\right)
\left(\frac{b^a}{\Gamma(a)} y^{a-1} e^{-by}\right)^\alpha;
\end{eqnarray}
here, \eqref{EQ:dpd5_simple} provides a closed-form objective function for computing the MDPDE under the gamma model. Subsequently, for convenience, we denote
\begin{equation}\label{eq:notation_r_M}
    r_{a,\alpha} := (1+\alpha)a-\alpha, \qquad M_{\bm\theta, \alpha} :=  \int
f_{\Gamma}^{1+\alpha}(y;\bm{\theta})dy = \frac{
b^\alpha
\Gamma\left(r_{a,\alpha}\right)
}
{
\Gamma(a)^{1+\alpha}
(1+\alpha)^{r_{a,\alpha}}
}.
\end{equation}
Here, naturally, $M_{\bm\theta, \alpha}$ is well-defined assuming $r_{a,\alpha} > 0$, i.e., $a > \alpha / (1+\alpha)$.

\section{Robustness and asymptotics} \label{sec:properties}
\noindent Robustness properties of the MDPDE for the gamma distribution follow directly from the fact that the estimator is an $M$-estimator, that is, it satisfies an estimating equation of the form
$\sum_{i=1}^{n}\bm{\psi}(Y_i,\bm{\theta})=\bm{0}$, where $\bm{\psi}(\cdot)$ is a vector-valued function having the same dimension as the parameter space. For the gamma distribution, the parameter vector is $\bm{\theta}=(a,b)^\top$, and hence $\bm{\psi}(\cdot)$ is two-dimensional. Unbiased estimating equations for any $\alpha\geq0$ can be obtained by differentiating $H_{\alpha,n}(\bm{\theta})$ in \eqref{eq:dpd5}. The corresponding estimating equations for $a$ and $b$ are given by
\begin{eqnarray}\label{eq:dpd6}
\nonumber
U_n(a;b)
&\equiv&
\frac{1}{n}\sum_{i=1}^{n}
u_{a;b}(Y_i)
f_\Gamma^\alpha(Y_i;\bm{\theta})
- \int u_{a;b}(y)
f_\Gamma^{1+\alpha}(y;\bm{\theta}) dy
=0,
\\
U_n(b;a)
&\equiv&
\frac{1}{n}\sum_{i=1}^{n}
u_{b;a}(Y_i)
f_\Gamma^\alpha(Y_i;\bm{\theta})
-
\int
u_{b;a}(y)
f_\Gamma^{1+\alpha}(y;\bm{\theta}) dy
=0,
\end{eqnarray}
where $u_{a;b}(y) = \frac{\partial}{\partial a} \log f_\Gamma(y;\bm{\theta})$, and
$ u_{b;a}(y) =\frac{\partial}{\partial b} \log f_\Gamma(y;\bm{\theta})$ denote the score functions. At $\alpha=0$, \eqref{eq:dpd6} reduces to the usual estimating equations associated with MLE. For any $\alpha>0$, the MDPDE produces weighted score equations with weights $f_\Gamma^\alpha(Y_i;\bm{\theta})$. Consequently, observations lying in low-density regions of the fitted gamma distribution receive smaller weights, thereby reducing their influence on the resulting parameter estimates and yielding robustness against outliers.

Further simplification of \eqref{eq:dpd6} is possible by deriving explicit expressions for the score vector and the associated integral terms. For the gamma distribution, following \eqref{eq:gamma_density}, the score vector
$\bm{u}_{\bm{\theta}}(y)=\left(u_{a;b}(y),u_{b;a}(y)\right)^\top$
is given by
\begin{equation}\label{eq:score_vector}
\bm{u}_{\bm{\theta}}(y) = (\log b-\psi(a)+\log y, ~ a/b-y)^\top,
\end{equation}
where $\psi(\cdot)$ denotes the digamma function. Using \eqref{eq:score_vector} and the notations in \eqref{eq:notation_r_M}, the estimating equations in \eqref{eq:dpd6} can be written as
\begin{eqnarray}
\nonumber
U_n(a;b) &=& \frac{1}{n} \sum_{i=1}^{n}
\left(\log b-\psi(a)+\log Y_i\right)
\left(\frac{b^a}{\Gamma(a)} Y_i^{a-1}\exp\{-bY_i\}\right)^\alpha
\\
&&
- M_{\bm \theta, \alpha}
\Big[\psi(r_{a, \alpha}) - \psi(a) -\log(1+\alpha)\Big], \\
\nonumber
U_n(b;a) &=& \frac{1}{n} \sum_{i=1}^{n}
\left(\frac{a}{b}-Y_i\right)
\left(\frac{b^a}{\Gamma(a)} Y_i^{a-1}
\exp\{-bY_i\} \right)^\alpha
- \frac{\alpha}{(1+\alpha)b} M_{\bm \theta, \alpha}.
\end{eqnarray}
The MDPDE $\widehat{\bm{\theta}}_\alpha=(\widehat a_\alpha,\widehat b_\alpha)^\top$ is obtained by solving the above system of nonlinear equations simultaneously.

\subsection{Asymptotic relative efficiency}\label{subsec:are}

\noindent We now investigate the asymptotic properties of the estimator $\widehat{\bm{\theta}}_\alpha$ defined in \eqref{eq:dpd4}. Closed-form expressions for $\widehat{\bm{\theta}}_\alpha$ are not available, even in the special case of the MLE corresponding to $\alpha=0$. Consequently, explicit expressions for the finite-sample mean and covariance matrix of $\widehat{\bm{\theta}}_\alpha$ are also unavailable. Nevertheless, the large-sample properties of the estimator can be established using the general asymptotic theory of MDPDE developed by \cite{basu1998robust}. Suppose that the gamma model is correctly specified so that the true data-generating distribution is $G=F_\Gamma(\cdot;\bm{\theta}_\ast)$,
where $\bm{\theta}_\ast=(a_\ast,b_\ast)^\top\in\Theta$ denotes the true parameter vector. Under the regularity conditions given in \cite{basu1998robust}, the estimator
$\widehat{\bm{\theta}}_\alpha~=~(\widehat a_\alpha,\widehat b_\alpha)^\top$ is a consistent estimator of $\bm{\theta}_\ast$, and the asymptotic distribution of
$\widetilde{\bm{\theta}}_{\alpha,n} =
\sqrt{n}(\widehat{\bm{\theta}}_\alpha -~\bm{\theta}_\ast)$
is bivariate normal with a mean vector
$\bm{\mu}_{\widetilde{\bm{\theta}}} = (0,0)^\top$
and covariance matrix
$\bm{\Sigma}_{\widetilde{\bm{\theta}}}
= \bm{J}_\alpha(\bm{\theta}_\ast)^{-1}
\bm{K}_\alpha(\bm{\theta}_\ast)
\bm{J}_\alpha(\bm{\theta}_\ast)^{-1}$,
where
\begin{eqnarray}\label{eq:J_K_xi}
&&
\bm{J}_\alpha(\bm{\theta}) = \int \bm{u}_{\bm{\theta}}(y)
\bm{u}_{\bm{\theta}}(y)^\top
f_\Gamma^{1+\alpha}(y;\bm{\theta}) dy, \\
\nonumber && \bm{K}_\alpha(\bm{\theta}) = \bm{J}_{2\alpha}(\bm{\theta}) - 
\bm{\xi}_\alpha(\bm{\theta}) \bm{\xi}_\alpha(\bm{\theta})^\top,
~
\textrm{where} ~
\bm{\xi}_\alpha(\bm{\theta}) = \int
\bm{u}_{\bm{\theta}}(y) f_\Gamma^{1+\alpha}(y;\bm{\theta})dy.
\end{eqnarray}
Here, $\bm{u}_{\bm{\theta}}(y)$ denotes the score vector associated with the gamma distribution as defined before. The matrices $\bm{J}_\alpha(\bm{\theta})$ and $\bm{K}_\alpha(\bm{\theta})$ are commonly referred to as the sensitivity matrix and the variability matrix, respectively, and together determine the asymptotic covariance matrix of the proposed estimator through the familiar sandwich form.

Further simplification of the quantities in \eqref{eq:J_K_xi} yields explicit expressions for the vector $\bm{\xi}_\alpha(\bm{\theta})$ and the matrix $\bm{J}_\alpha(\bm{\theta})$. Let $\bm{J}_\alpha^{(i,j)}(\bm{\theta})$ denote the $(i,j)$-th element of the symmetric matrix $\bm{J}_\alpha(\bm{\theta})$, $i,j=1,2$. Using the score vector in \eqref{eq:score_vector} together with the identity
$ M_{\bm{\theta},\alpha} = \int f_\Gamma^{1+\alpha}(y;\bm{\theta})\,dy$, closed-form expressions for $\bm{\xi}_\alpha(\bm{\theta})$ and the elements of $\bm{J}_\alpha(\bm{\theta})$ can be obtained after straightforward algebra and are 
\begin{eqnarray}\label{eq:J_K_xi_simple}
&&
\bm{\xi}_{\alpha}(\bm{\theta}) =
M_{\bm{\theta},\alpha} \left[\psi(r_{a,\alpha})-\psi(a)-\log(1+\alpha), ~ \dfrac{\alpha}{(1+\alpha)b} \right]^\top,
\nonumber\\
&&
\bm{J}^{(1,1)}_\alpha(\bm{\theta})
= M_{\bm{\theta},\alpha}
\left[\psi'(r_{a,\alpha}) +
\left\{\psi(r_{a,\alpha}) -\psi(a) -\log(1+\alpha) \right\}^2
\right],
\nonumber\\
&&
\bm{J}^{(1,2)}_\alpha(\bm{\theta})
=
\bm{J}^{(2,1)}_\alpha(\bm{\theta})
=
\frac{M_{\bm{\theta},\alpha}}{b(1+\alpha)}
\left[
\alpha \left\{ \psi(r_{a,\alpha}) -\psi(a) -\log(1+\alpha)
\right\} - 1 \right],
\nonumber\\
&&
\bm{J}^{(2,2)}_\alpha(\bm{\theta})
=\frac{M_{\bm\theta,\alpha}}
{(1+\alpha)^2b^2} \left\{r_{a, \alpha}+\alpha^2\right\},
\end{eqnarray}
where $\psi'(\cdot)$ denotes the trigamma function and $r_{a,\alpha}$ is as in \eqref{eq:notation_r_M}.

The elements of the asymptotic covariance matrix $\bm{\Sigma}_{\tilde{\bm{\theta}}}$ depend on the model parameters $a$ and $b$, as well as on the MDPDE tuning parameter $\alpha$. The condition for obtaining the closed-form expressions of $\bm{\xi}_\alpha(\bm{\theta})$ and $\bm{J}_\alpha(\bm{\theta})$ is $r_{a,\alpha} > 0$. Besides, $\bm{\Sigma}_{\widetilde{\bm{\theta}}}$ depends on $\bm{K}_\alpha(\bm{\theta})$, which eventually depends on $\bm{J}_{2\alpha}(\bm{\theta})$. Thus, obtaining the closed-form expression of $\bm{\Sigma}_{\widetilde{\bm{\theta}}}$ based on \eqref{eq:J_K_xi_simple} requires $r_{a,2\alpha} > 0$, i.e., $a > 2\alpha/(1+2\alpha)$. To illustrate these relationships, we plot the elements $\bm{\Sigma}^{(i,j)}_{\tilde{\bm{\theta}}}$, $i,j=1,2$, for a range of values of $a$, $b$, and $\alpha$ in Figure \ref{fig_cov_comparison}. The diagonal elements $\bm{\Sigma}^{(1,1)}_{\tilde{\bm{\theta}}}$ and $\bm{\Sigma}^{(2,2)}_{\tilde{\bm{\theta}}}$ represent the asymptotic variances of $\sqrt{n}\widehat{a}_{\alpha}$ and $\sqrt{n}\widehat{b}_{\alpha}$, respectively, while the off-diagonal element $\bm{\Sigma}^{(1,2)}_{\tilde{\bm{\theta}}}$ (equal to $\bm{\Sigma}^{(2,1)}_{\tilde{\bm{\theta}}}$) represents the asymptotic covariance between $\sqrt{n}\widehat{a}_{\alpha}$ and $\sqrt{n}\widehat{b}_{\alpha}$. We observe that the asymptotic variance of $\sqrt{n}\widehat{a}_{\alpha}$ appears to be insensitive to the rate parameter $b$ but varies considerably with the shape parameter $a$. In contrast, the asymptotic variance of $\sqrt{n}\widehat{b}_{\alpha}$ increases rapidly with $b$, while the asymptotic covariance between $\sqrt{n}\widehat{a}_{\alpha}$ and $\sqrt{n}\widehat{b}_{\alpha}$ appears to exhibit a linear increase with $b$. As expected, increasing the tuning parameter from $\alpha=0.1$ to $\alpha=0.5$ leads to a noticeable increase in all three quantities, reflecting the loss in asymptotic efficiency associated with improved robustness. Motivated by these empirical observations, we theoretically establish the following proportionality relationships.

\begin{figure}[t]
    \centering
\includegraphics[width = \linewidth]{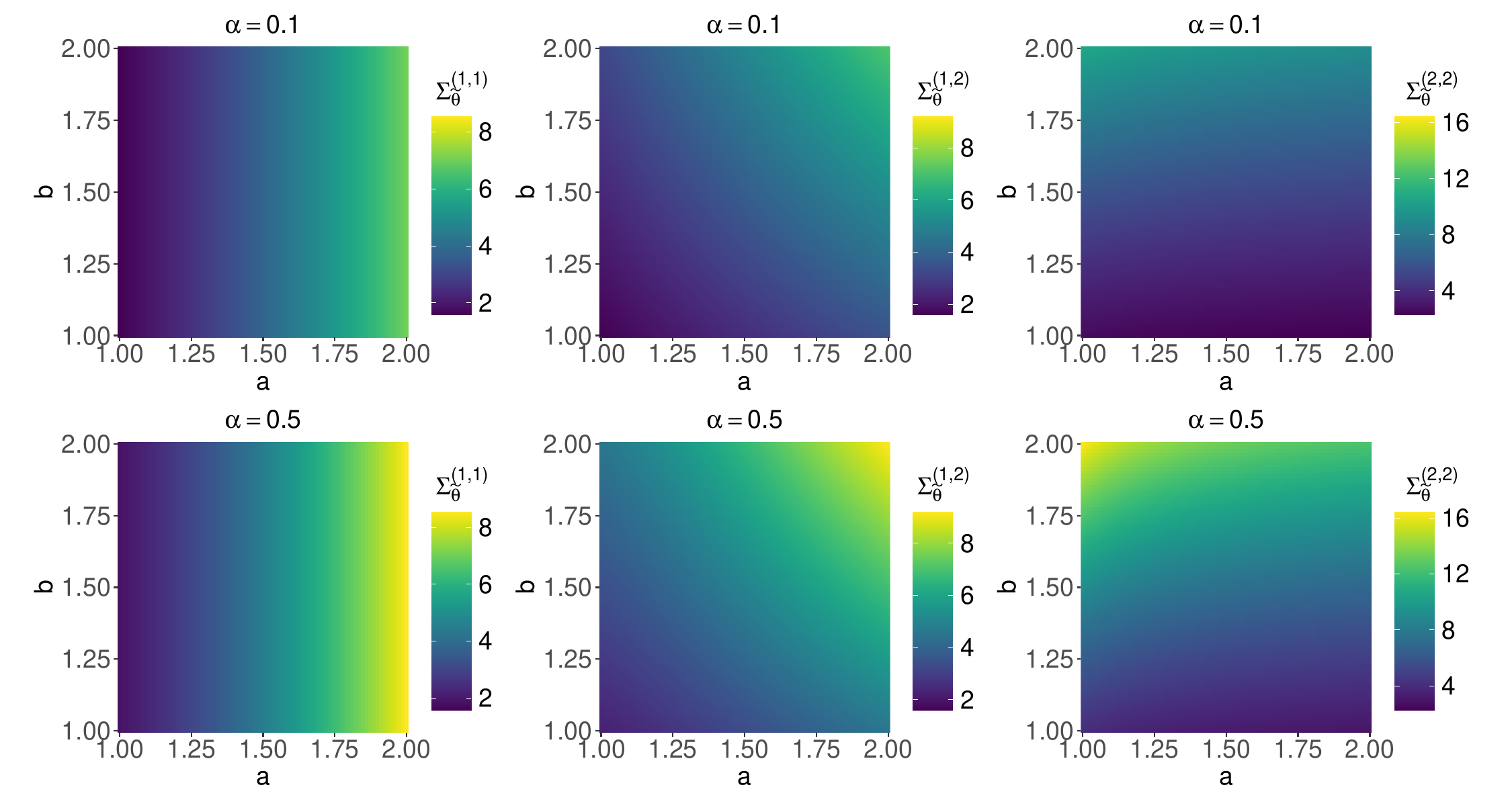}
    \caption{Variation of the elements of the asymptotic covariance matrix $\bm{\Sigma}_{\tilde{\bm{\theta}}}$ as functions of the shape parameter $a$, the rate parameter $b$, and the MDPDE tuning parameter $\alpha$. Here, $\bm{\Sigma}_{\tilde{\bm{\theta}}}^{(i,j)}$ denotes the $(i,j)$-th element of $\bm{\Sigma}_{\tilde{\bm{\theta}}}$, $i,j=1,2$.}
    \label{fig_cov_comparison}
    \vspace{-4mm}
\end{figure}

\begin{theorem}\label{thm1}
For the asymptotic covariance matrix $\bm{\Sigma}_{\widetilde{\bm{\theta}}}$ of the rescaled estimator $\widehat{\bm{\theta}}_{\alpha}$, the relationships between the elements of $\bm{\Sigma}_{\widetilde{\bm{\theta}}}$ and the rate parameter $b$ are given by: (i) $\bm{\Sigma}^{(1,1)}_{\widetilde{\bm{\theta}}}$ is independent of $b$, (ii) $\bm{\Sigma}^{(1,2)}_{\widetilde{\bm{\theta}}}\propto b$, and (iii) $\bm{\Sigma}^{(2,2)}_{\widetilde{\bm{\theta}}}\propto b^2$.
\end{theorem}


The proof of Theorem~\ref{thm1} is provided in the Appendix. One can verify that the asymptotic covariance matrix of the MDPDE is minimized at $\alpha=0$, which corresponds to MLE. Since the MLE is asymptotically the most efficient member of the MDPDE family under the correctly specified model, it is natural to investigate the asymptotic relative efficiency (ARE) of the proposed estimator. For each parameter, the ARE is defined as the ratio of the asymptotic variance of the MLE to that of the corresponding MDPDE, assuming that the data are generated from the gamma model without contamination. By definition, the ARE equals one at $\alpha=0$. Although the matrices $\bm{J}_\alpha(\bm{\theta})$, $\bm{K}_\alpha(\bm{\theta})$, and consequently $\bm{\Sigma}_{\widetilde{\bm{\theta}}}$ admit closed-form expressions, they involve the gamma, digamma, and trigamma functions, making further analytical simplification difficult. Nevertheless, Theorem \ref{thm1} immediately leads to the following result regarding the dependence of the ARE on $b$.

\begin{theorem}\label{thm2}
The asymptotic relative efficiencies of both $\widehat a_\alpha$ and $\widehat b_\alpha$ are independent of the rate parameter $b$.
\vspace{-1mm}
\begin{proof}
From Theorem \ref{thm1}, the asymptotic variances satisfy
$\bm{\Sigma}^{(1,1)}_{\widetilde{\bm{\theta}}} = L_{1,1}(\bm{\phi})$ and $\bm{\Sigma}^{(2,2)}_{\widetilde{\bm{\theta}}}
= b^2 L_{2,2}(\bm{\phi})$, where $\bm{\phi}=(a,\alpha)'$, and $L_{1,1}(\cdot)$ and $L_{2,2}(\cdot)$ are free of $b$. Therefore, $\mathrm{ARE}(\widehat a_\alpha) = L_{1,1}((a,0)^\top) / L_{1,1}((a,\alpha)^\top)$, which is free of $b$. Likewise, $$\mathrm{ARE}(\widehat b_\alpha) = b^2L_{2,2}((a,0)^\top) / (b^2L_{2,2}((a,\alpha)^\top)) = L_{2,2}((a,0)^\top) / L_{2,2}((a,\alpha)^\top),$$ which is also free of $b$.
\end{proof}
\end{theorem}

\begin{table}[t]
\begin{center}
\caption{\normalsize{Variation of $\mathrm{ARE}(\widehat a_\alpha)$ and $\mathrm{ARE}(\widehat b_\alpha)$ for different choices of $a$ and $\alpha$. Entries marked ``--'' correspond to parameter combinations for which the regularity condition $a>2\alpha/(1+2\alpha)$ is violated.}}
\label{tab_are_comparison}
\begin{tabular}{|c|ccccc|ccccc|}
\hline
& \multicolumn{5}{c|}{$\mathrm{ARE}(\widehat a_\alpha)$}
& \multicolumn{5}{c|}{$\mathrm{ARE}(\widehat b_\alpha)$}\\
\hline
\diagbox{$\alpha$}{$a$}
& 0.5 & 1 & 2 & 5 & 10
& 0.5 & 1 & 2 & 5 & 10\\
\hline
0.10 & 0.92 & 0.98 & 0.98 & 0.98 & 0.98 & 0.91 & 0.96 & 0.97 & 0.97 & 0.98\\
0.20 & 0.70 & 0.94 & 0.95 & 0.93 & 0.93 & 0.66 & 0.88 & 0.91 & 0.92 & 0.92\\
0.30 & 0.38 & 0.90 & 0.90 & 0.87 & 0.86 & 0.30 & 0.78 & 0.84 & 0.85 & 0.85\\
0.40 & 0.09 & 0.86 & 0.85 & 0.82 & 0.80 & 0.04 & 0.70 & 0.78 & 0.79 & 0.79\\
0.50 & -- & 0.82 & 0.81 & 0.76 & 0.75 & -- & 0.62 & 0.72 & 0.73 & 0.73\\
0.60 & -- & 0.78 & 0.77 & 0.71 & 0.69 & -- & 0.56 & 0.67 & 0.68 & 0.68\\
0.70 & -- & 0.75 & 0.73 & 0.67 & 0.65 & -- & 0.51 & 0.62 & 0.63 & 0.63\\
0.80 & -- & 0.72 & 0.70 & 0.63 & 0.61 & -- & 0.48 & 0.59 & 0.60 & 0.60\\
0.90 & -- & 0.70 & 0.67 & 0.60 & 0.58 & -- & 0.44 & 0.56 & 0.57 & 0.57\\
1.00 & -- & 0.68 & 0.65 & 0.57 & 0.56 & -- & 0.42 & 0.54 & 0.54 & 0.54\\
\hline
\end{tabular}
\end{center}
\vspace{-4mm}
\end{table}
The proof of Theorem~\ref{thm2} is provided in the Appendix. Table \ref{tab_are_comparison} reports $\mathrm{ARE}(\widehat a_\alpha)$ and $\mathrm{ARE}(\widehat b_\alpha)$ for a range of values of the shape parameter $a$ and the tuning parameter $\alpha$. As expected, the asymptotic relative efficiencies decrease as $\alpha$ increases, reflecting the classical trade-off between robustness and efficiency. However, except when $a=0.5$, for moderate values of $\alpha$, the reduction in efficiency is relatively small, indicating that the MDPDE retains much of the efficiency of the MLE while providing substantially improved robustness. Consequently, except in cases with evidence that the shape parameter is less than one, when the underlying data may contain atypical observations, the modest loss in asymptotic efficiency is often outweighed by the increased stability of the resulting estimators. The numerical results further illustrate how the efficiencies of $\widehat a_\alpha$ and $\widehat b_\alpha$ vary with the shape parameter $a$, providing useful guidance for selecting an appropriate tuning parameter value in practical applications.

\subsection{Influence function analysis}\label{subsec:influence}
\noindent We now investigate the robustness properties of the proposed estimator through the influence function (IF), originally introduced by \cite{Hampeletc:1986}. Let
$ G_{\pi,y}=(1-\pi)G+\pi\delta_y $
denote the contaminated distribution, where $G$ is the true underlying distribution, $\pi$ is the contamination proportion, and $\delta_y$ is the degenerate distribution placing unit mass at the contamination point $y$. For the MDPDE functional $\bm{T}_\alpha(\cdot)$ defined in (\ref{eq:dpd2}), the quantity
$ \bm{T}_\alpha(G_{\pi,y})-\bm{T}_\alpha(G) $
represents the asymptotic bias induced by contamination. The influence function measures the infinitesimal effect of contamination on the estimator and is defined by
\begin{eqnarray}\label{eq:if_def}
\mathrm{IF}(y;\bm{T}_\alpha,G)
= \lim_{\pi\rightarrow0} \frac{\bm{T}_\alpha(G_{\pi,y}) -
\bm{T}_\alpha(G)}{\pi}.
\end{eqnarray}

Since the gamma distribution is parameterized by the two-dimensional parameter vector $\bm{\theta}=(a,b)^\top$, the influence function is a mapping from $\mathbb{R}^{+}$ to $\mathbb{R}^{2}$. If $\mathrm{IF}(y;\bm{T}_\alpha, G)$ is unbounded, an arbitrarily small amount of contamination can produce an arbitrarily large asymptotic bias, indicating that the corresponding estimator is non-robust. Conversely, if the influence function is bounded over the support of the distribution, the asymptotic bias remains bounded even when the contamination occurs at an extreme observation. Therefore, boundedness of the influence function provides an important theoretical measure of robustness for the proposed estimator.

Under the correctly specified gamma model, namely
$G=F_{\Gamma}(\cdot;\bm{\theta}_{\ast})$, where
$\bm{\theta}_{\ast} = (a_{\ast},b_{\ast})^\top \in\Theta$,
the influence function of the MDPDE follows directly from \cite{basu1998robust} and is given by
\begin{eqnarray}\label{eq:if_mdpde}
 &&
\mathrm{IF}
\left(
y;
\bm{T}_\alpha,
F_{\Gamma}(\cdot;\bm{\theta}_{\ast})
\right)
=
\bm{J}_\alpha(\bm{\theta}_{\ast})^{-1}
\left[
\bm{u}_{\bm{\theta}_{\ast}}(y)
f_{\Gamma}^{\alpha}(y;\bm{\theta}_{\ast})
- \bm{\xi}_{\alpha}(\bm{\theta}_\ast)
\right],
\end{eqnarray}
where $\bm{u}_{\bm{\theta}_{\ast}}(\cdot)$ denotes the score vector given in \eqref{eq:score_vector}, while $\bm{J}_\alpha(\cdot)$ is defined in \eqref{eq:J_K_xi} and its explicit form is given in \eqref{eq:J_K_xi_simple}.

\begin{theorem}\label{thm3}
Suppose the true parameter vector
$\bm{\theta}_{\ast}=(a_{\ast},b_{\ast})^\top$
is finite and satisfies $a_{\ast}> 1$, and thus $a_{\ast}> \alpha / (1 + \alpha)$ and all expressions in (\ref{eq:J_K_xi_simple}) are well defined. Then, each component of the influence function
$\mathrm{IF}\left(y;\bm{T}_\alpha, F_{\Gamma}(\cdot;\bm{\theta}_{\ast})\right)$ is bounded if and only if $\alpha>0$.
\end{theorem}

\begin{figure}[]
	\centering
	\includegraphics[width = \linewidth]{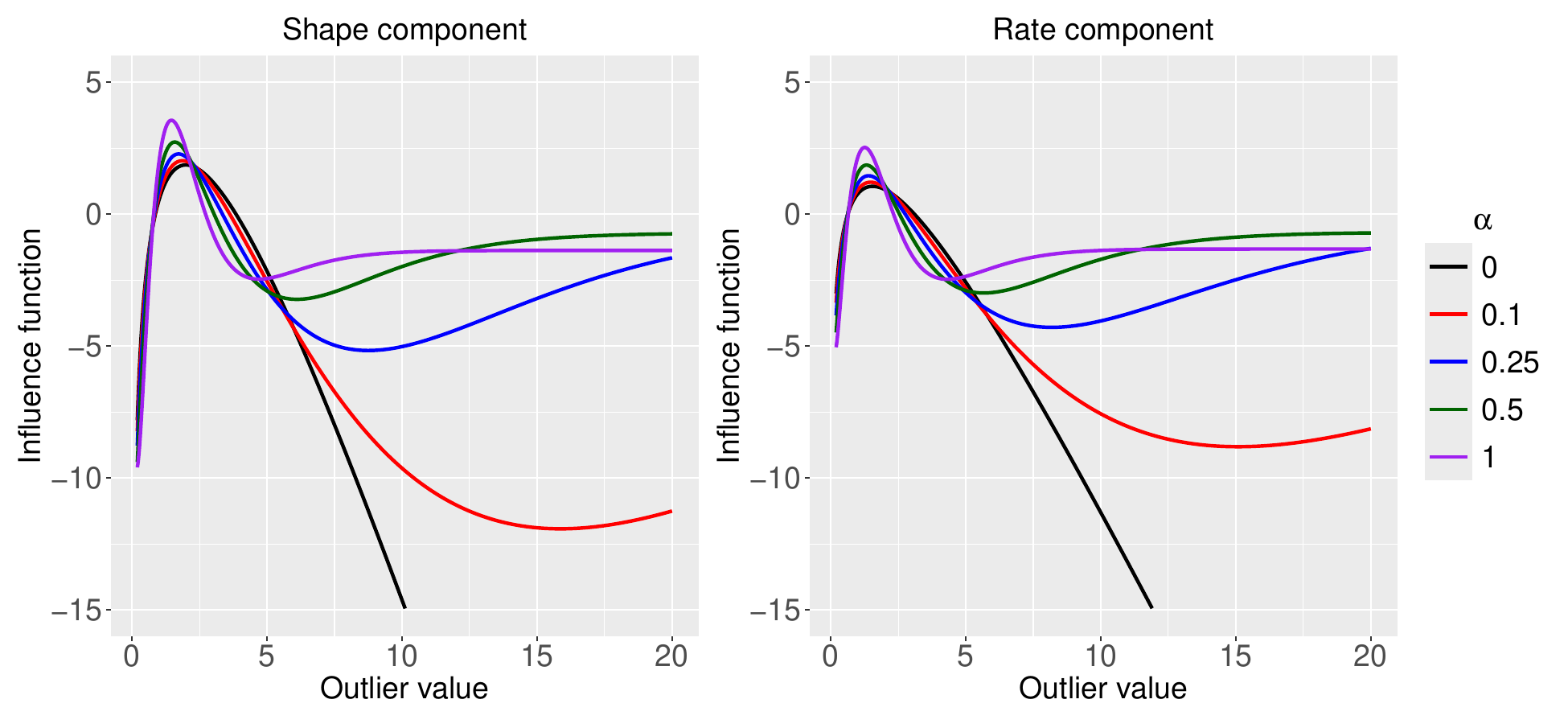}
\caption{Components of the influence function $\mathrm{IF}
\left(y; \bm{T}_\alpha, F_{\Gamma}(\cdot;\bm{\theta}_{\ast})
\right)$ for $a_\ast$ (left panel) and $b_\ast$ (right panel) at $\textrm{Gamma}(2, 1)$ distribution.}
	\label{fig_if}
    \vspace{-2mm}
\end{figure}

The proof of Theorem~\ref{thm3} is provided in the Appendix. In case the condition $a_\ast > 1$ for Theorem~\ref{thm3} does not hold, i.e., $a_\ast \leq 1$, but $a_{\ast}> \alpha / (1 + \alpha)$ holds, the determining term $\lim_{y\rightarrow0^+}  y^{\alpha(a_\ast-1)} \log y = -\infty$. Hence, the boundedness is not guaranteed. Besides, when $a_\ast < 1$, the ARE drops drastically as $\alpha$ increases, as seen in Table~\ref{tab_are_comparison}. We hence do not explore the case of $a_\ast \leq 1$ further. To illustrate the theoretical findings in Theorem~\ref{thm3}, Figure~\ref{fig_if} displays the two components of
$\mathrm{IF}\left(y;\bm{T}_\alpha, F_{\Gamma}(\cdot;\bm{\theta}_{\ast})\right)
$
for several choices of the tuning parameter $\alpha$ under a gamma distribution with true parameter values $a_{\ast}=2$ and $b_{\ast}=1$. As predicted by the theorem, both components are unbounded when $\alpha=0$, corresponding to the maximum likelihood estimator, whereas they remain bounded for every $\alpha>0$. Moreover, the influence functions decay more rapidly toward zero as $\alpha$ increases, indicating that larger values of the tuning parameter provide stronger robustness against extreme observations. These results confirm the robustness of the MDPDE while illustrating the non-robustness of MLE.

\subsection{Optimal tuning parameter selection}
\label{subsec:optimal_alpha}

\noindent The tuning parameter $\alpha$ controls the balance between robustness and statistical efficiency in the MDPDE. Smaller values of $\alpha$ yield estimators that are highly efficient under the assumed gamma model but less robust to outliers, whereas larger values yield greater robustness at the expense of efficiency. Consequently, selecting an appropriate value of $\alpha$ is an important practical issue. Since the degree of contamination in a dataset is generally unknown, a fully data-driven procedure is desirable for determining the optimal tuning parameter.

Following the approach of \cite{fujisawa2006robust}, subsequently adopted by \cite{seo2017extreme}, \cite{hazra2024robust}, and \cite{hazra2025minimum}, we select the tuning parameter by minimizing an empirical Cram\'er--von Mises (CVM) criterion. Specifically, the optimal tuning parameter is defined as
\begin{equation}\label{EQ:CVM}
\alpha_{\mathrm{opt}}
=
\arg\min_{\alpha}
\frac{1}{n}
\sum_{i=1}^{n}
\left\{
\frac{i}{n+1}
-
F_{\Gamma}
\left(
Y_{(i)};
\widehat{\bm{\theta}}_{\alpha}^{(-i)}
\right)
\right\}^{2},
\end{equation}
where $Y_{(1)},\ldots,Y_{(n)}$ denote the order statistics of the sample and $\widehat{\bm{\theta}}_{\alpha}^{(-i)}$ is the MDPDE computed after deleting the $i$th observation from the dataset. This leave-one-out criterion favors values of $\alpha$ that produce fitted gamma distributions with good predictive agreement while remaining resistant to the influence of individual observations. Alternative approaches for selecting the tuning parameter have also been proposed in the literature \citep{hong2001automatic, warwick2005choosing, basak2021optimal}, several of which are applicable to general $M$-estimators rather than being specific to the MDPDE framework.

\section{Simulation studies}\label{sec:simulation}

\noindent To assess the performance of the proposed MDPDE for different choices of the tuning parameter $\alpha$, including the data-driven choice $\alpha_{\mathrm{opt}}$ defined in (\ref{EQ:CVM}), and to compare it with some existing estimation procedures, we conduct a Monte Carlo simulation study under four different contamination levels. We first discuss some potential competing estimation procedures in Section~\ref{subsec:gamma_parest}.

\subsection{Alternative parameter estimation procedures}
\label{subsec:gamma_parest}

\noindent Several methods have been proposed in the literature for estimating the parameters of univariate probability distributions, including the two-parameter gamma distribution \citep{johnson1994continuous, hosking1990moments, stuart2010kendall}. Let $\mathcal{Y}=\{Y_1,\ldots, Y_n\}$ denote an independent and identically distributed random sample from a $\mathrm{Gamma}(a,b)$ distribution, where $a$ and $b$ denote the shape and rate parameters, respectively.

The most widely used estimation procedure is the maximum likelihood (ML) estimation. The maximum likelihood estimators (MLEs) are obtained by maximizing the log-likelihood function
$
\ell(a,b;\mathcal{Y}) = \sum_{i=1}^{n} \log f_{\Gamma}(Y_i;a,b)
$
over the parameter space. Although closed-form expressions for both estimators are not available, the likelihood equations can be simplified considerably. Given the MLE $\widehat a_{\mathrm{ML}}$ of the shape parameter, the MLE of the rate parameter has the explicit form
$\widehat b_{\mathrm{ML}} = \widehat a_{\mathrm{ML}} / \overline{Y}$, where $\overline{Y}$ denotes the sample mean. Consequently, the estimation problem reduces to solving a single nonlinear equation in $a$, after which $\widehat b_{\mathrm{ML}}$ is obtained directly.

The method of moments (MM) estimates are obtained by equating the sample mean and sample variance to their corresponding population moments. Since $\mathrm{E}(Y)=a / b$ and $\mathrm{Var}(Y)=a/b^2$, the resulting estimators are
$\widehat a_{\mathrm{MM}} = \overline{Y}^{\,2} / S^2$, and $\widehat b_{\mathrm{MM}} = \overline{Y} / S^2$, where $S^2$ denotes the sample variance.

The percentile (PT) estimation method minimizes the discrepancy between the ordered observations and the theoretical quantiles of the fitted gamma distribution. Let $Q_{\Gamma}(p;a,b) = F_{\Gamma}^{-1}(p;a,b),~ 0<p<1,$ denote the gamma quantile function, and let $Y_{(1)}<\cdots<Y_{(n)}$ be the ordered sample. The percentile estimators are obtained by minimizing
$ \sum_{i=1}^{n} \left[Y_{(i)} - Q_{\Gamma}\left(i/(n+1); a, b \right) \right]^2.$

The least squares (LS) estimation method instead minimizes the discrepancy between the empirical and fitted distribution functions. Specifically, the estimators are obtained by minimizing $\sum_{i=1}^{n} \left[ i/(n+1) - F_{\Gamma}
(Y_{(i)}; a, b) \right]^2.$ A weighted version of this approach, known as the weighted least squares (WLS) estimator, minimizes
$\sum_{i=1}^{n} w_i \left[i/(n+1) - F_{\Gamma}(Y_{(i)};a,b)
\right]^2$, where $w_i =\{(n+1)^2(n+2) \} / \{i(n-i+1) \}$. The weighting scheme assigns greater emphasis to observations in the tails of the distribution.

In $L$-moment (LM) estimation, instead of obtaining parameter estimates by equating the raw or central population moments with their sample counterparts as in the method of moments, we equate the population and sample $L$-moments, which are linear combinations of order statistics. Since $L$-moments are based on order statistics rather than ordinary moments, they are generally regarded as more robust to atypical observations than the traditional method of moments. For the gamma distribution, let $\lambda_1(a,b)$ and $\lambda_2(a,b)$ denote the first and second population $L$-moments, respectively. The LM estimators are obtained by solving the estimating equations
\begin{eqnarray}
\nonumber \lambda_1(a,b) &=& n^{-1} \sum_{i=1}^n Y_{(i)} = \Bar{Y}, \\
\nonumber \lambda_2(a,b) &=& \frac{2}{n(n-1)} \sum_{i=1}^n (i-1)Y_{(i)} - \Bar{Y},
\end{eqnarray}
where $Y_{(1)} < \cdots < Y_{(n)}$ denote the order statistics. The population $L$-moments $\lambda_1(a,b)$ and $\lambda_2(a,b)$ are evaluated numerically from the fitted gamma distribution. Solving the above system yields the LM estimators $\widehat a_{\mathrm{LM}}$ and $\widehat b_{\mathrm{LM}}$ of the shape and rate parameters, respectively. The $L$-moment estimation for the gamma distribution is implemented in the \texttt{pelgam} function within the \texttt{R} package \texttt{lmom} \citep{hosking2026}.

We further compare the performance of the above approaches with the proposed MDPDE through simulation studies.

\subsection{Simulation designs and results}
\label{subsec:sim_designs}

\noindent We conduct a Monte Carlo simulation study under four different contamination levels: 0\%, 1\%, 5\%, and 10\%. Specifically, we generate 1000 random samples of size $n=100$ from the $\textrm{Gamma}(2,1)$ distribution, which serves as the true data-generating model, and replace a specified proportion of observations with two types of outliers.

In the first contamination setting (C1), the outliers are generated from a degenerate distribution concentrated at the $(1 - 10^{-3})$-th quantile of the $\textrm{Gamma}(2,1)$ distribution, i.e., at a value of 9.233. Such observations may occasionally arise in rainfall records due to exceptionally rare meteorological events. Although these values are important in extreme value analysis, they can substantially distort inference when the primary objective is to model the central part of the rainfall distribution, such as estimating the 30\%-probability rainfall amount, i.e., the $0.7^{th}$ quantile of the true rainfall distribution. In the second contamination setting (C2), the outliers are generated from a degenerate distribution at the $(1-10^{-5})$-th quantile of the $\textrm{Gamma}(2,1)$ distribution, i.e., at a value of 14.237. Observations of this magnitude are extremely unlikely to occur in practice within a sample of size 100 and are more representative of gross recording errors or instrument malfunction. These two contamination schemes therefore represent moderate and severe departures from the assumed gamma model. For each simulated dataset, the model parameters are estimated using the MDPDE for several values of $\alpha$, along with the competing estimation methods described in Section \ref{subsec:gamma_parest}. The empirical bias and mean squared error (MSE) of the resulting estimators are then computed over the 1000 replications. The simulation results for contamination settings C1 and C2 are reported in Tables \ref{table_sim1} and \ref{table_sim2}, respectively.

\begin{table}[!tb]
\caption{\normalsize Average bias and mean square error (MSE) in the estimation of the gamma parameters $a$ and $b$ using different estimation methods under different outlier percentages. Here, the outlier value is the $(1-10^{-3})$-th quantile of the true distribution $\mathrm{Gamma}(a,b)$. Highlighted entries in each column correspond to the cases with bias and MSE closest to zero.}
\begin{tabular}{cccccccccc}
\hline
 &  & \multicolumn{8}{c}{Outlier percentage} \\
 &  & \multicolumn{2}{c}{0\%} & \multicolumn{2}{c}{1\%} & \multicolumn{2}{c}{5\%} & \multicolumn{2}{c}{10\%} \\
Method & $\bm{\theta}$ & Bias & MSE & Bias & MSE & Bias & MSE & Bias & MSE \\
\hline
\multirow{2}{*}{ML}
& $a$ & 0.073 & \textbf{0.087} & -0.078 & \textbf{0.061} & -0.381 & 0.173 & -0.587 & 0.363 \\
& $b$ & 0.041 & 0.028 & -0.069 & \textbf{0.021} & -0.312 & 0.102 & -0.482 & 0.234 \\\cline{2-10}

\multirow{2}{*}{MM}
& $a$ & 0.090 & 0.131 & -0.264 & 0.123 & -0.731 & 0.552 & -0.862 & 0.757 \\
& $b$ & 0.049 & 0.039 & -0.161 & 0.038 & -0.462 & 0.215 & -0.584 & 0.342 \\\cline{2-10}

\multirow{2}{*}{PT}
& $a$ & -0.070 & 0.142 & -0.461 & 0.272 & -0.907 & 0.840 & -0.919 & 0.859 \\
& $b$ & -0.041 & 0.039 & -0.260 & 0.080 & -0.536 & 0.289 & -0.610 & 0.373 \\\cline{2-10}

\multirow{2}{*}{LS}
& $a$ & \textbf{0.026} & 0.110 & -0.045 & 0.104 & -0.183 & 0.123 & -0.391 & 0.229 \\
& $b$ & 0.012 & 0.036 & -0.036 & 0.035 & -0.158 & 0.052 & -0.325 & 0.127 \\\cline{2-10}

\multirow{2}{*}{WLS}
& $a$ & 0.040 & 0.094 & -0.042 & 0.083 & -0.190 & 0.108 & -0.387 & 0.216 \\
& $b$ & 0.021 & 0.031 & -0.036 & 0.029 & -0.169 & 0.051 & -0.334 & 0.130 \\\cline{2-10}

\multirow{2}{*}{LM}
& $a$ & 0.051 & 0.094 & -0.162 & 0.087 & -0.566 & 0.347 & -0.790 & 0.641 \\
& $b$ & \textbf{-0.003} & \textbf{0.027} & 0.147 & 0.049 & 0.660 & 0.466 & 1.273 & 1.658 \\\cline{2-10}

\multirow{2}{*}{\begin{tabular}[c]{@{}c@{}}MDPDE\\ ($\alpha=0.1$)\end{tabular}}
& $a$ & 0.070 & 0.088 & -0.047 & 0.065 & -0.319 & 0.137 & -0.549 & 0.323 \\
& $b$ & 0.039 & 0.028 & -0.045 & 0.022 & -0.265 & 0.078 & -0.453 & 0.209 \\\cline{2-10}

\multirow{2}{*}{\begin{tabular}[c]{@{}c@{}}MDPDE\\ ($\alpha=0.2$)\end{tabular}}
& $a$ & 0.070 & 0.091 & -0.020 & 0.071 & -0.243 & 0.107 & -0.491 & 0.268 \\
& $b$ & 0.039 & 0.030 & -0.023 & 0.025 & -0.207 & 0.057 & -0.409 & 0.174 \\\cline{2-10}

\multirow{2}{*}{\begin{tabular}[c]{@{}c@{}}MDPDE\\ ($\alpha=0.5$)\end{tabular}}
& $a$ & 0.075 & 0.107 & 0.032 & 0.094 & -0.049 & 0.096 & -0.225 & 0.142 \\
& $b$ & 0.043 & 0.038 & 0.019 & 0.036 & -0.056 & \textbf{0.039} & -0.204 & 0.078 \\\cline{2-10}

\multirow{2}{*}{\begin{tabular}[c]{@{}c@{}}MDPDE\\ ($\alpha=1$)\end{tabular}}
& $a$ & 0.091 & 0.140 & 0.065 & 0.133 & \textbf{0.021} & 0.125 & \textbf{-0.090} & \textbf{0.125} \\
& $b$ & 0.054 & 0.054 & 0.044 & 0.055 & \textbf{-0.005} & 0.048 & \textbf{-0.104} & \textbf{0.055} \\\cline{2-10}

\multirow{2}{*}{\begin{tabular}[c]{@{}c@{}}MDPDE\\ ($\alpha=\alpha_{\mathrm{opt}}$)\end{tabular}}
& $a$ & 0.056 & 0.095 & \textbf{-0.010} & 0.079 & -0.174 & \textbf{0.092} & -0.370 & 0.181 \\
& $b$ & 0.028 & 0.033 & \textbf{-0.017} & 0.028 & -0.156 & 0.044 & -0.318 & 0.114 \\\cline{2-10}
\hline
\end{tabular}
\label{table_sim1}
\end{table}

We first summarize the findings based on Table \ref{table_sim1}. Under the uncontaminated setting (0\% outliers), the LS estimator produces the smallest bias for the shape parameter $a$, whereas the LM estimator yields the smallest bias and MSE for the rate parameter $b$. The ML estimator achieves the smallest MSE for $a$, while MDPDE with $\alpha=0.1$ performs comparably well. As the contamination level increases to 1\%, MDPDE with $\alpha_{\mathrm{opt}}$ yields the smallest biases for both parameters, whereas ML continues to provide the smallest MSE. Under moderate contamination (5\%), the superiority of robust estimation becomes evident: MDPDE with $\alpha=1$ gives the smallest biases for both parameters, MDPDE with $\alpha=\alpha_{\mathrm{opt}}$ achieves the smallest MSE for $a$, and MDPDE with $\alpha=0.5$ achieves the smallest MSE for $b$. Although the LM estimator performs well in the absence of contamination, its performance deteriorates rapidly once outliers are introduced, particularly for estimating $b$. Under severe contamination (10\%), MDPDE with $\alpha=1$ attains both the smallest biases and the smallest MSEs for $a$ and $b$, while MDPDE with $\alpha=0.5$ remains a competitive alternative. Overall, the results demonstrate that MDPDE offers a favorable balance between robustness and efficiency, and the data-adaptive choice $\alpha=\alpha_{\mathrm{opt}}$ provides a practical compromise when the contamination level is unknown.

\begin{table}[!tb]
\caption{\normalsize Average bias and mean square error (MSE) in the estimation of the gamma parameters $a$ and $b$ using different estimation methods under different outlier percentages. Here, the outlier value is the $(1-10^{-5})$-th quantile of the true distribution $\mathrm{Gamma}(a,b)$. Highlighted entries in each column correspond to the cases with bias and MSE closest to zero.}
\begin{tabular}{cccccccccc}
\hline
 &  & \multicolumn{8}{c}{Outlier percentage} \\
 &  & \multicolumn{2}{c}{0\%} & \multicolumn{2}{c}{1\%} & \multicolumn{2}{c}{5\%} & \multicolumn{2}{c}{10\%} \\
Method & $\bm{\theta}$ & Bias & MSE & Bias & MSE & Bias & MSE & Bias & MSE \\
\hline

\multirow{2}{*}{ML}
& $a$ & 0.046 & 0.077 & -0.193 & 0.081 & -0.670 & 0.463 & -0.899 & 0.815 \\
& $b$ & 0.030 & \textbf{0.025} & -0.148 & 0.034 & -0.490 & 0.242 & -0.659 & 0.434 \\\cline{2-10}

\multirow{2}{*}{MM}
& $a$ & 0.057 & 0.123 & -0.688 & 0.495 & -1.248 & 1.564 & -1.322 & 1.751 \\
& $b$ & 0.036 & 0.036 & -0.384 & 0.150 & -0.712 & 0.508 & -0.790 & 0.625 \\\cline{2-10}

\multirow{2}{*}{PT}
& $a$ & -0.102 & 0.149 & -0.832 & 0.728 & -1.436 & 2.067 & -1.386 & 1.925 \\
& $b$ & -0.055 & 0.041 & -0.445 & 0.204 & -0.774 & 0.599 & -0.810 & 0.656 \\\cline{2-10}

\multirow{2}{*}{LS}
& $a$ & -0.019 & 0.096 & -0.040 & 0.110 & -0.192 & 0.125 & -0.352 & 0.190 \\
& $b$ & -0.007 & 0.032 & -0.038 & 0.038 & -0.164 & 0.054 & -0.305 & 0.111 \\\cline{2-10}

\multirow{2}{*}{WLS}
& $a$ & \textbf{0.004} & 0.082 & -0.035 & 0.090 & -0.188 & 0.107 & -0.329 & 0.167 \\
& $b$ & \textbf{0.006} & 0.027 & -0.036 & 0.031 & -0.168 & 0.051 & -0.301 & 0.108 \\\cline{2-10}

\multirow{2}{*}{LM}
& $a$ & 0.010 & 0.087 & -0.299 & 0.139 & -0.910 & 0.841 & -1.152 & 1.332 \\
& $b$ & 0.014 & 0.027 & 0.272 & 0.106 & 1.414 & 2.047 & 2.824 & 8.034 \\\cline{2-10}

\multirow{2}{*}{\begin{tabular}[c]{@{}c@{}}MDPDE\\ ($\alpha=0.1$)\end{tabular}}
& $a$ & 0.041 & \textbf{0.077} & -0.074 & \textbf{0.072} & -0.489 & 0.266 & -0.796 & 0.643 \\
& $b$ & 0.028 & 0.025 & -0.064 & \textbf{0.026} & -0.367 & 0.142 & -0.596 & 0.357 \\\cline{2-10}

\multirow{2}{*}{\begin{tabular}[c]{@{}c@{}}MDPDE\\ ($\alpha=0.2$)\end{tabular}}
& $a$ & 0.040 & 0.079 & \textbf{-0.007} & 0.081 & -0.251 & 0.130 & -0.607 & 0.396 \\
& $b$ & 0.028 & 0.027 & \textbf{-0.015} & 0.029 & -0.197 & 0.063 & -0.468 & 0.228 \\\cline{2-10}

\multirow{2}{*}{\begin{tabular}[c]{@{}c@{}}MDPDE\\ ($\alpha=0.5$)\end{tabular}}
& $a$ & 0.047 & 0.092 & 0.044 & 0.104 & \textbf{0.003} & \textbf{0.101} & -0.060 & \textbf{0.104} \\
& $b$ & 0.036 & 0.035 & 0.022 & 0.039 & \textbf{-0.011} & \textbf{0.038} & \textbf{-0.074} & \textbf{0.044} \\\cline{2-10}

\multirow{2}{*}{\begin{tabular}[c]{@{}c@{}}MDPDE\\ ($\alpha=1$)\end{tabular}}
& $a$ & 0.067 & 0.125 & 0.064 & 0.143 & 0.005 & 0.123 & \textbf{-0.059} & 0.107 \\
& $b$ & 0.054 & 0.052 & 0.036 & 0.055 & -0.014 & 0.047 & -0.086 & 0.046 \\\cline{2-10}

\multirow{2}{*}{\begin{tabular}[c]{@{}c@{}}MDPDE\\ ($\alpha=\alpha_{\mathrm{opt}}$)\end{tabular}}
& $a$ & 0.028 & 0.082 & -0.022 & 0.088 & -0.237 & 0.115 & -0.432 & 0.222 \\
& $b$ & 0.019 & 0.029 & -0.031 & 0.032 & -0.188 & 0.055 & -0.344 & 0.129 \\\cline{2-10}

\hline
\end{tabular}
\label{table_sim2}
\end{table}

We next discuss the findings based on Table \ref{table_sim2}, when the contaminating observations are substantially more extreme. Under the uncontaminated setting (0\% outliers), the WLS estimator yields the smallest bias for both $a$ and $b$, while the ML estimator attains the smallest MSE for $b$. MDPDE with $\alpha=0.1$ achieves the same minimum MSE for $a$ as ML, indicating only a negligible loss of efficiency in the absence of contamination. With 1\% contamination, MDPDE with $\alpha=0.2$ provides the smallest biases for both parameters, whereas MDPDE with $\alpha=0.1$ achieves the smallest MSEs. As the contamination level increases to 5\%, the advantages of robust estimation become much more pronounced. MDPDE with $\alpha=0.5$ produces the smallest bias and MSE for both parameters, while the performances of the classical methods, particularly MM, PT, and LM, deteriorate substantially. For 10\% contamination, MDPDE continues to dominate the competing methods. MDPDE with $\alpha=1$ yields the smallest bias for the shape parameter $a$, whereas MDPDE with $\alpha=0.5$ gives the smallest bias for the rate parameter $b$ and also attains the smallest MSEs for both parameters. Although the data-adaptive choice $\alpha=\alpha_{\mathrm{opt}}$ is not uniformly optimal under such severe contamination, it consistently outperforms classical estimation methods in both bias and MSE. Overall, the results demonstrate that the benefits of MDPDE become increasingly evident as the severity of contamination increases, with moderate values of the tuning parameter ($\alpha\approx0.5$) providing an excellent compromise between robustness and efficiency.

Overall, the results in Tables \ref{table_sim1} and \ref{table_sim2} demonstrate the advantages of MDPDE over the classical estimation methods discussed in Section \ref{subsec:gamma_parest}, particularly in the presence of contamination. When the contamination level is low (1\%) and the outliers are relatively mild (Table \ref{table_sim1}), the ML estimator still achieves the smallest MSE, although MDPDE with positive values of $\alpha$ often yields substantially smaller bias at only a modest efficiency loss. For more extreme outliers (Table \ref{table_sim2}), even at the same contamination level, MDPDE with $\alpha=0.1$, $0.2$, or $0.5$ consistently outperforms the competing methods in terms of bias and MSE. While the influence function analysis in Section \ref{subsec:influence} establishes the bounded influence of infinitesimal contamination for $\alpha>0$, the simulation studies illustrate the finite-sample benefits of MDPDE under moderate and severe contamination. As the contamination percentage increases, the performance of the ML, MM, PT, LS, WLS, and particularly the LM estimators deteriorates rapidly, whereas MDPDE remains comparatively stable. For moderate and high contamination levels (5\% and 10\%), MDPDE with $\alpha=0.5$ or $1$ generally provides the best overall performance, while the data-adaptive choice $\alpha=\alpha_{\mathrm{opt}}$ consistently performs well and substantially outperforms the classical estimation methods. Overall, these results suggest that MDPDE offers an effective balance between robustness and efficiency, with $\alpha_{\mathrm{opt}}$ serving as a practical choice when the extent of contamination is unknown.

\section{Indian rainfall data analysis} \label{sec:data_application}

\noindent Unlike administrative state boundaries, India is partitioned into 36 meteorological subdivisions that are relatively homogeneous from a climatological perspective \citep{guhathakurta2008trends}. The geographical boundaries of these subdivisions are shown in Figure \ref{fig_data_illstration}. We obtained areally weighted rainfall data for these subdivisions from the Open Government Data (OGD) Platform, India (\url{https://data.gov.in}), along with the corresponding subdivision boundary shapefile. The dataset was prepared by the India Meteorological Department (IMD), which first estimated monthly rainfall for each of the 641 districts by averaging observations from all available rain-gauge stations within each district. Subsequently, subdivision-level rainfall was computed by applying an area-weighted aggregation of the district-level rainfall values \cite{guhathakurta2008trends, guhathakurta2011new}. In the present study, we consider total monsoon rainfall observed over the four monsoon months (June--September) for each of the 36 meteorological subdivisions during the period 1951--2014. This results in a dataset consisting of 64 annual monsoon rainfall totals for each meteorological subdivision. We illustrate the data for the years 1951 and 2014 in Figure \ref{fig_data_illstration}. A similar month-level dataset has been analyzed by \cite{hazra2024robust} for robust rainfall modeling and further analyzed by \cite{bhowmik2026bayesian} for Bayesian spatiotemporal modeling. 

\begin{figure}[t]
    \centering
\includegraphics[width = 0.9\linewidth]{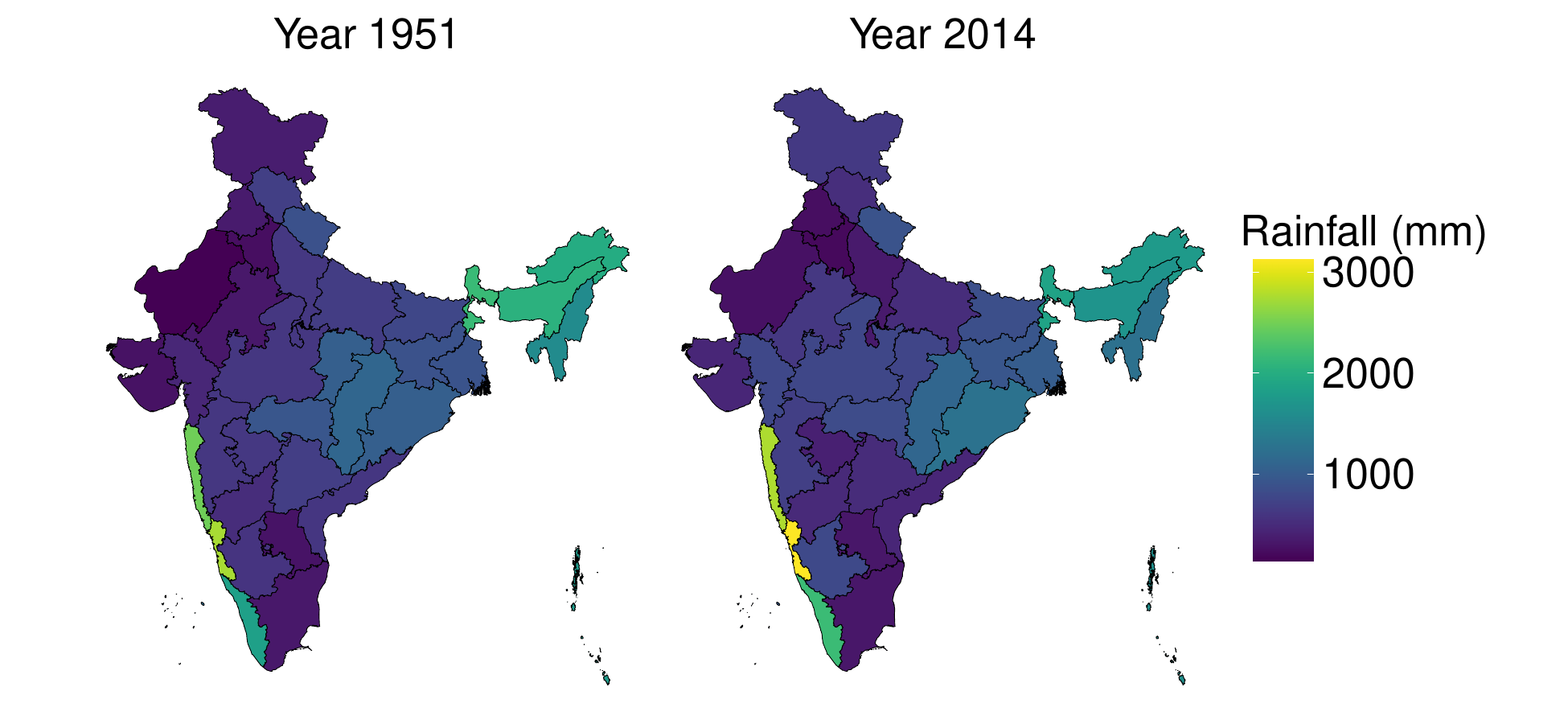}
    \caption{Meteorological subdivision-wise rainfall data for 36 subdivisions of India during the monsoon season (June--September) for the years 1951 and 2014.}
    \label{fig_data_illstration}
\end{figure}

Since the dataset spans 64 years (1951--2014), a period during which substantial climatic changes have occurred, the presence of temporal trends in the annual monsoon rainfall totals is plausible. In an earlier spatiotemporal Bayesian model-based analysis, \cite{bhowmik2026bayesian} identify a significant trend in the marginal distribution for three out of 36 subdivisions. Prior to fitting the gamma distribution, we therefore investigate the existence of monotonic trends using the robust nonparametric Cox--Stuart trend test \citep{cox1955some}, implemented through the function \texttt{cs.test} in the \texttt{R} package \texttt{trend} \citep{pohlert2023trend}. The test is applied separately to the annual total monsoon rainfall series for each of the 36 meteorological subdivisions. Significant trends are detected for two subdivisions: (i) Nagaland, Manipur, Mizoram, and Tripura, and (ii) East Uttar Pradesh, at the 1\% significance levels. For each subdivision, we estimate the temporal trend by fitting an $L_1$ regression model to the log-transformed monsoon rainfall totals; i.e., we fit $\log \widetilde{Y}_i = \beta_0 + \beta_1 \text{year}_i + \varepsilon_i, ~i=1,\ldots,n$, estimate $\beta_1$ by $\widehat{\beta}_1$ using $L_1$ regression, and obtain $Y_i = \exp\{ \log \widetilde{Y}_i - \widehat{\beta}_1 \text{year}_i \}, ~i=1,\ldots,n$. Here, $\widetilde{Y}_i$ denotes the raw data for the $i^{th}$ year, while $Y_i$ denotes the processed dataset to be analyzed using the methodology in this paper. The estimated $\widehat{\beta}_1$s range between -0.005 and 0.007, across subdivisions. Throughout the remainder of the paper, the term \textit{detrended data} refers to $Y_i,i=1,\ldots, n$. This data processing step removes large-scale temporal variation while preserving positive support and a similar scale across observations; Pearson's correlation between the raw and detrended data is 0.997. Applying the Cox--Stuart test to the detrended data confirms the absence of significant trends, allowing us to reasonably assume that the observations within each meteorological subdivision are independent and identically distributed and hence suitable for modeling using the gamma distribution and the proposed MDPDE methodology. Here, the independence assumption is justified because weather conditions during different monsoon seasons are likely independent.

In the context of robust estimation for rainfall modeling, \cite{seo2017extreme} classify an observation as an outlier if it lies more than 1.5 interquartile ranges above the third quartile or below the first quartile. However, rainfall data are typically positively skewed, making the conventional boxplot rule prone to misclassifying extreme but genuine observations. The Adjusted-Boxplot method accounts for skewness by employing a robust measure of skewness that is itself resistant to outliers when determining the whisker lengths. Following \cite{hazra2024robust}, we identify outliers in the detrended annual monsoon rainfall totals using the widely adopted Adjusted-Boxplot procedure of \cite{hubert2008adjusted}, implemented in the \texttt{R} package \texttt{univOutl} \citep{univOutl}. For each meteorological subdivision, we compute the proportion of observations classified as outliers. The highest proportion (10.9\%) is observed in the Orissa and Chhattisgarh subdivisions, while no outliers are detected in 11 subdivisions. Among the remaining subdivisions, 11 contain one outlier, 6 contain two outliers, 2 contain four outliers, and 4 contain five outliers. These results indicate that outlying observations are not uncommon in the dataset and motivate the use of robust estimation methods such as the proposed MDPDE.

\begin{figure}[t]
    \centering
\includegraphics[width = 0.7\linewidth]{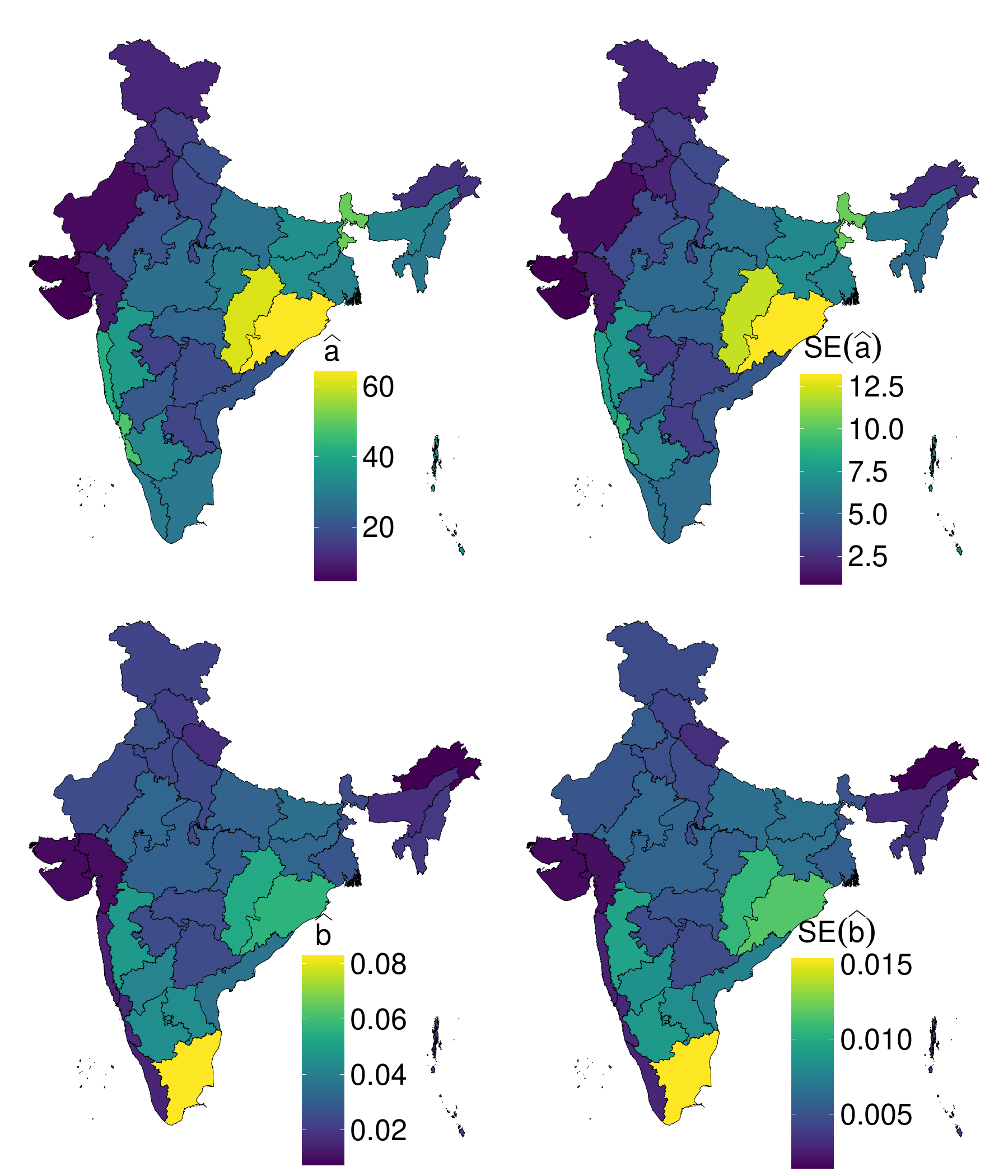}
    \caption{Meteorological subdivision-wise minimum density power divergence estimates of the gamma shape ($a$) and rate ($b$) parameters and their corresponding asymptotic standard errors calculated based on \eqref{eq:J_K_xi} and tuning parameter $\alpha$ selected according to \eqref{EQ:CVM}.}
    \label{mdpde_maps_a_b}
    \vspace{-4mm}
\end{figure}

Figure \ref{mdpde_maps_a_b} presents the MDPDE estimates of the gamma shape and rate parameters, together with their corresponding standard errors, for the 36 meteorological subdivisions of India. The tuning parameter of the MDPDE is selected separately for each subdivision using the data-driven procedure described in Section \ref{subsec:optimal_alpha}, yielding optimal values of $\alpha$ ranging from 0.127 to 0.499, which indicates that moderate robustness is preferred across most regions. The estimated shape parameter varies considerably, from 4.821 for Saurashtra, Kutch, and Diu to 64.198 for Orissa, with other high values observed for Chhattisgarh (60.971), Sub-Himalayan West Bengal and Sikkim (50.560), and Coastal Karnataka (47.585). In contrast, arid regions such as West Rajasthan (6.593) and Gujarat Region, Daman and Diu, and Nagar Haveli (8.870) exhibit the smallest shape parameter estimates. The estimated rate parameter follows a similar spatial pattern, ranging from 0.007 for Arunachal Pradesh to 0.083 for Tamil Nadu and Pondicherry. The standard errors also exhibit substantial spatial variation and generally increase with the corresponding parameter estimates. The standard error of the estimated shape parameter varies considerably, from 0.871 for Saurashtra, Kutch, and Diu to 13.199 for Orissa. The standard error of the estimated rate parameter also varies significantly, ranging from 0.001 for Arunachal Pradesh to 0.015 for Tamil Nadu and Pondicherry. When compared with the alternative maximum likelihood estimation, the MLEs of the shape and rate parameters exhibit a spatial pattern similar to that of MDPDE. However, the differences between the shape parameter estimates from MDPDE and MLE range from -6.544 to 18.115, and for the rate parameter, they range from -0.005 to 0.016. Despite differences in estimates, the median (across 36 subdivisions) asymptotic relative efficiency values for the two parameters are 0.840 and 0.850, respectively, indicating minimal compromise in estimation uncertainty.

\begin{figure}[!htbp]
    \centering
\includegraphics[width = 0.7\linewidth]{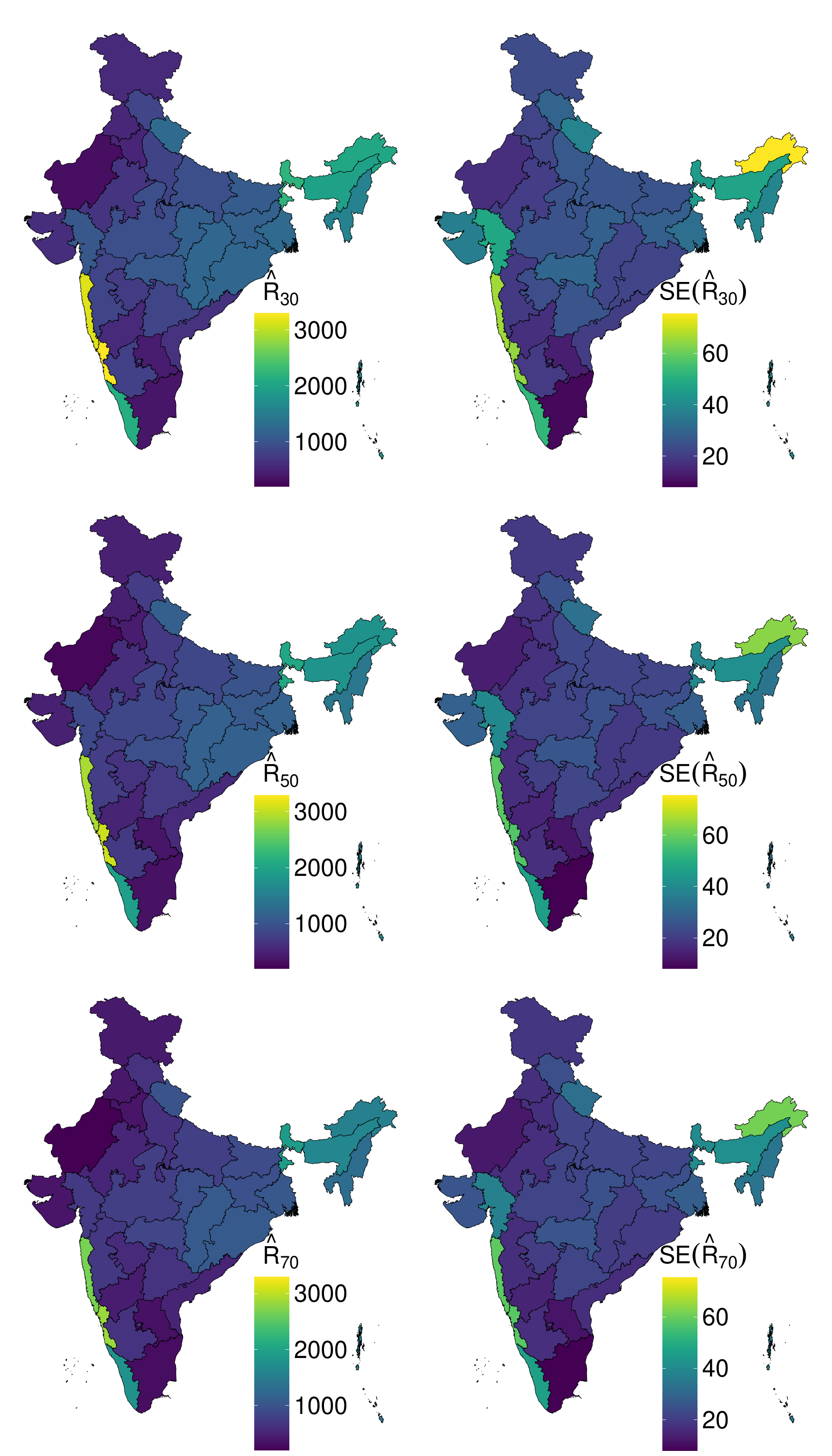}
    \caption{Meteorological subdivision-wise estimates (in the left column) for 30\%, 50\%, and 70\% probability rainfall amounts ($R_{30}$, $R_{50}$, and $R_{50}$, respectively) for the monsoon season (June--September), calculated based on minimum density power divergence estimates of the gamma shape ($a$) and rate ($b$) parameters. Their corresponding asymptotic standard errors (in the right column) are calculated using the delta method.}
    \label{mdpde_maps_probrain}
    \vspace{-4mm}
\end{figure}

Figure \ref{mdpde_maps_probrain} displays the estimated monsoon rainfall totals corresponding to 30\%, 50\%, and 70\% exceedance probabilities, defined as the $(1-q/100)$th quantile of the fitted gamma distribution for $q=30, 50$, and $70$, respectively. The corresponding standard errors are obtained using the delta method, where the gradients of the gamma quantile function with respect to $(a,b)$ are evaluated numerically using the \texttt{grad} function from the \texttt{numDeriv R} package. The estimated rainfall amounts exhibit pronounced spatial variability, with the highest values consistently observed for Coastal Karnataka (3287, 3049, and 2822 mm, respectively), followed by Konkan and Goa (3132, 2890, and 2660 mm), whereas the lowest values occur in West Rajasthan (311, 254, and 205 mm) and Tamil Nadu and Pondicherry (371, 337, and 304 mm). As expected, the estimated rainfall totals decrease monotonically as the exceedance probability increases from 30\% to 70\% across all subdivisions. The corresponding standard errors range from 8 mm to 75 mm across all probability levels and display a spatial pattern similar to that of the rainfall estimates, with the largest uncertainties occurring over the high-rainfall coastal and northeastern subdivisions and the smallest uncertainties over the relatively dry regions of western and southern India. The estimated rainfall amounts based on MLE show a spatial pattern similar to those calculated using MDPDE. However, the differences in estimated rainfall amounts using MDPDE and MLE vary across subdivisions, ranging from -61 mm to 21 mm for 30\% probability rainfall, from -46 mm to 47 mm for 50\% probability rainfall, and from -34 mm to 68 mm for 70\% probability rainfall. The ratio of the asymptotic variances of the estimated rainfall amounts using MDPDE and MLE remains close to one; across subdivisions, it ranges from 0.761 to 1.430 for the 30\% probability rainfall, from 0.775 to 1.364 for the 50\% probability rainfall, and from 0.804 to 1.346 for the 70\% probability rainfall. Unlike the asymptotic variance of the MDPDE estimates being always larger than that based on MLE, the asymptotic variance of the final rainfall estimates does not exhibit any specific ordering.

Together, these figures demonstrate that the proposed robust gamma model provides stable estimates of both distributional parameters and practically relevant rainfall quantiles, along with reliable measures of their estimation uncertainty.

\section{Discussions and conclusions}\label{sec:conclusion}

In this paper, we develop a robust estimation framework for the two-parameter gamma distribution based on the minimum density power divergence estimator (MDPDE). The proposed methodology provides a flexible alternative to the classical maximum likelihood estimator by introducing a tuning parameter that balances statistical efficiency under the assumed model and robustness against data contamination. We derive the estimating equations for the gamma distribution, establish the asymptotic distribution of the proposed estimator, obtain explicit expressions for its asymptotic covariance matrix, and investigate its robustness through influence function analysis. The asymptotic relative efficiency of the estimator is also studied for different choices of the tuning parameter, illustrating the expected trade-off between robustness and efficiency.

A comprehensive simulation study demonstrates that the proposed estimator performs competitively with several existing estimation methods, including maximum likelihood, method of moments, percentile, least squares, weighted least squares, and $L$-moment estimation. In the absence of contamination, the MDPDE incurs only a small loss of efficiency compared with maximum-likelihood estimation. However, as the proportion or severity of contamination increases, the robust estimator substantially outperforms the competing methods in both bias and mean squared error. The data-driven procedure for selecting the tuning parameter consistently produces estimators that maintain a desirable balance between robustness and efficiency across a wide range of contamination scenarios.

The methodology is further used to analyze total southwest monsoon rainfall data from the 36 meteorological subdivisions of India. After removing temporal trends, the proposed robust estimator is used to fit gamma distributions separately for each subdivision. The resulting parameter estimates and estimated rainfall quantiles exhibit considerable spatial heterogeneity, reflecting the diverse climatological characteristics across India. The estimated standard errors, obtained from the asymptotic covariance matrix and the delta method, provide useful measures of uncertainty for both the model parameters and practically relevant rainfall quantiles.

Although the proposed methodology focuses on the two-parameter gamma distribution, the underlying density power divergence framework is considerably more general. The derivations presented here may serve as a basis for developing robust estimation procedures for more flexible gamma-based models that frequently arise in hydrological and environmental applications \citep{chen2017entropy, ilinca2023flood}. Another natural extension is to incorporate covariate information through gamma regression models, where robust estimation could improve inference in the presence of anomalous observations. From a computational perspective, efficient implementations for large-scale spatial or spatiotemporal datasets also constitute an important direction for future research \citep{basu2026high, bhowmik2026bayesian}.

Overall, the proposed MDPDE provides a computationally feasible, theoretically justified, and practically effective robust estimation procedure for the gamma distribution. The combination of robustness, high efficiency under the assumed model, and a principled, data-driven tuning-parameter selection strategy makes it an attractive alternative to classical estimation methods for analyzing positively skewed data that may contain outliers.


\section*{Appendix}
\setcounter{theorem}{0}
\begin{theorem}
For the asymptotic covariance matrix $\bm{\Sigma}_{\widetilde{\bm{\theta}}}$ of the rescaled estimator $\widehat{\bm{\theta}}_{\alpha}$, the relationships between the elements of $\bm{\Sigma}_{\widetilde{\bm{\theta}}}$ and the rate parameter $b$ are given by: (i) $\bm{\Sigma}^{(1,1)}_{\widetilde{\bm{\theta}}}$ is independent of $b$, (ii) $\bm{\Sigma}^{(1,2)}_{\widetilde{\bm{\theta}}}\propto b$, and (iii) $\bm{\Sigma}^{(2,2)}_{\widetilde{\bm{\theta}}}\propto b^2$.
\end{theorem}
\begin{proof}
Using the expressions in \eqref{eq:J_K_xi_simple} together with \eqref{eq:J_K_xi} and the expression of $M_{\bm \theta, \alpha} \propto b^\alpha$ in \eqref{eq:notation_r_M}, the matrices $\bm{J}_\alpha(\bm{\theta})$ and $\bm{K}_\alpha(\bm{\theta})$ can be expressed as
\begin{equation}
\nonumber
\bm{J}_\alpha(\bm{\theta})
=
\begin{pmatrix}
b^{\alpha}J_{1,1}(\bm{\phi})
&
b^{\alpha-1}J_{1,2}(\bm{\phi})
\\[3mm]
b^{\alpha-1}J_{2,1}(\bm{\phi})
&
b^{\alpha-2}J_{2,2}(\bm{\phi})
\end{pmatrix},
~~
\bm{K}_\alpha(\bm{\theta})
=
\begin{pmatrix}
b^{2\alpha}K_{1,1}(\bm{\phi})
&
b^{2\alpha-1}K_{1,2}(\bm{\phi})
\\[3mm]
b^{2\alpha-1}K_{2,1}(\bm{\phi})
&
b^{2\alpha-2}K_{2,2}(\bm{\phi})
\end{pmatrix},
\end{equation}
where $J_{i,j}(\cdot)$ and $K_{i,j}(\cdot)$ are functions of $\bm{\phi}=(a,\alpha)^\top$ only and are therefore free of $b$. Straightforward algebra then shows that the elements of
$\bm{\Sigma}_{\widetilde{\bm{\theta}}}
~=~
\bm{J}_\alpha(\bm{\theta})^{-1}
\bm{K}_\alpha(\bm{\theta})
\bm{J}_\alpha(\bm{\theta})^{-1}
$ can be written as $\bm{\Sigma}^{(1,1)}_{\widetilde{\bm{\theta}}} =L_{1,1}(\bm{\phi})$, $\bm{\Sigma}^{(1,2)}_{\widetilde{\bm{\theta}}}
= b\,L_{1,2}(\bm{\phi})$, and $\bm{\Sigma}^{(2,2)}_{\widetilde{\bm{\theta}}} =
b^2L_{2,2}(\bm{\phi})$, where $L_{i,j}(\cdot)$ are functions of $\bm{\phi}$ that are free of $b$; this establishes the stated proportionality relations.
\end{proof}

\begin{theorem}\label{thm2}
The asymptotic relative efficiencies of both $\widehat a_\alpha$ and $\widehat b_\alpha$ are independent of the rate parameter $b$.
\vspace{-1mm}
\begin{proof}
From Theorem \ref{thm1}, the asymptotic variances satisfy
$\bm{\Sigma}^{(1,1)}_{\widetilde{\bm{\theta}}} = L_{1,1}(\bm{\phi})$ and $\bm{\Sigma}^{(2,2)}_{\widetilde{\bm{\theta}}}
= b^2 L_{2,2}(\bm{\phi})$, where $\bm{\phi}=(a,\alpha)'$, and $L_{1,1}(\cdot)$ and $L_{2,2}(\cdot)$ are free of $b$. Therefore, $\mathrm{ARE}(\widehat a_\alpha) = L_{1,1}((a,0)^\top) / L_{1,1}((a,\alpha)^\top)$, which is free of $b$. Likewise, $$\mathrm{ARE}(\widehat b_\alpha) = b^2L_{2,2}((a,0)^\top) / (b^2L_{2,2}((a,\alpha)^\top)) = L_{2,2}((a,0)^\top) / L_{2,2}((a,\alpha)^\top),$$ which is also free of $b$.
\end{proof}
\end{theorem}

\begin{theorem}\label{thm3}
Suppose the true parameter vector
$\bm{\theta}_{\ast}=(a_{\ast},b_{\ast})^\top$
is finite and satisfies $a_{\ast}> 1$, and thus $a_{\ast}> \alpha / (1 + \alpha)$ and all expressions in (\ref{eq:J_K_xi_simple}) are well defined. Then, each component of the influence function
$\mathrm{IF}\left(y;\bm{T}_\alpha, F_{\Gamma}(\cdot;\bm{\theta}_{\ast})\right)$ is bounded if and only if $\alpha>0$.
\end{theorem}

\begin{proof}
The matrix $\bm{J}_\alpha(\bm{\theta}_{\ast})^{-1}$
and the vector $\bm{\xi}_{\alpha}(\bm{\theta}_\ast)$
are independent of the contamination point $y$. Consequently, the boundedness of the influence function is determined entirely by $$\bm{u}_{\bm{\theta}_{\ast}}(y) f_{\Gamma}^{\alpha}(y;\bm{\theta}_{\ast}) = \left(\frac{b_\ast^{a_\ast}}{\Gamma(a_\ast)}\right)^{\alpha} y^{\alpha(a_\ast-1)}\exp\{-\alpha b_\ast y\} ~
\nonumber \times \begin{pmatrix} 
\log b_\ast-\psi(a_\ast)+\log y \\
a_\ast/b_\ast-y
\end{pmatrix}
.$$

To establish the ``only if'' part, consider $\alpha=0$. In this case, the term $f_{\Gamma}^{\alpha}(y;\bm{\theta}_{\ast})$ equals one, and the score vector is
$\bm{u}_{\bm{\theta}_\ast}(y) = \left(\log b_\ast-\psi(a_\ast)+\log y,\; a_\ast/b_\ast-y
\right)^\top.$
The first component diverges to $-\infty$ as $y\rightarrow0$, whereas the first and second components diverge to $\infty$ and $-\infty$, respectively, as $y\rightarrow\infty$. Hence, the score vector is unbounded, implying that the influence function is unbounded as well.

For the ``if'' part, suppose $\alpha>0$. The two components of
$\bm{u}_{\bm{\theta}_{\ast}}(y)f_{\Gamma}^{\alpha}(y;\bm{\theta}_{\ast})$ are continuous on $(0,\infty)$. Over any compact interval, they are therefore bounded by the extreme value theorem. Furthermore, $\lim_{y\rightarrow0^+} \bm{u}_{\bm{\theta}_{\ast}}(y) f_{\Gamma}^{\alpha}(y;\bm{\theta}_{\ast}) = \bm{0},$
provided $a_{\ast}>1$, since the polynomial decay of
$y^{\alpha(a_{\ast}-1)}$
dominates the logarithmic divergence of the first component of $\bm{u}_{\bm{\theta}_{\ast}}(y)$. Likewise,
$\lim_{y\rightarrow\infty}\bm{u}_{\bm{\theta}_{\ast}}(y)
f_{\Gamma}^{\alpha}(y;\bm{\theta}_{\ast}) =\bm{0},$
because the exponential factor $\exp(-\alpha b_{\ast}y)$
dominates the linear or logarithmic growth of the score vector components or the polynomial growth of $y^{\alpha(a_{\ast}-1)}$. Consequently,
$\bm{u}_{\bm{\theta}_{\ast}}(y)
f_{\Gamma}^{\alpha}(y;\bm{\theta}_{\ast})$
is bounded on $(0,\infty)$, establishing the boundedness of the influence function.
\end{proof}

\section*{Data and code availability}
The \texttt{R} code and data required to reproduce the figures and tables are accessible at the following GitHub repository: \url{https://github.com/arnabstatswithR/robustgamma}.

\section*{Conflict of interest}
The authors declare no competing interests.

\section*{Acknowledgement}
This work is supported by the Indian Institute of Technology Kanpur Initiation grant under Grant No. IITK/MATH/2021340. The author acknowledges the use of ChatGPT (OpenAI) and Grammarly for assistance with grammar correction, language refinement, and sentence restructuring. The scientific content, methodology, analysis, results, and conclusions presented in this paper are solely the responsibility of the author.

\bibliographystyle{plainnat} 
\bibliography{references}          
\end{document}